\documentclass[11pt]{article}

\usepackage{amsmath, amssymb, amsthm}
\usepackage{comment}
\usepackage{amsmath}
\usepackage{float}

\usepackage{graphicx}
\usepackage{pgf,tikz}
\usetikzlibrary{arrows,automata,positioning,calc}
\usetikzlibrary{decorations.pathmorphing}
\usepackage{algorithm}
\usepackage[noend]{algpseudocode}

\usepackage{epsfig}

\newcommand*{\dfa}{\textsc{dfa}}
\newcommand*{\dfas}{\textsc{dfa}s}
\newcommand*{\revdfa}{\textsc{rev-dfa}}
\newcommand*{\revdfas}{\textsc{rev-dfa}s}

\newcommand{\scc}{\textsc{scc}}
\newcommand{\sccs}{\textsc{scc}s}

\newcommand{\IN}{{\mbox{\rm\bf N}}}

\newtheorem{theorem}{Theorem}[section]
\newtheorem{lemma}[theorem]{Lemma}
\newtheorem{corollary}[theorem]{Corollary}
\newtheorem{ex}[theorem]{Example}

\begin{document}%
	\title{Minimal and Reduced Reversible Automata%
		\footnote{Preliminary version presented at
			\emph{DCFS 2016---\,Descriptional Complexity of Formal Systems},
			Bucharest, Romania, Jul 5--8, 2016
			[\emph{Lect.\ Notes Comput.\ Sci.},\ 9777, pp.\ 168--179, Springer, 2016]\@.%
		}%
	}
     
\author{Giovanna J. Lavado\and
Giovanni Pighizzini\and
Luca Prigioniero\and
\mbox{}\\
{\normalsize  Dipartimento di Informatica}\\
{\normalsize  Universit\`{a} degli Studi di Milano, Italy}\\
{\small\sf lavado@di.unimi.it} --
{\small\sf pighizzini@di.unimi.it} --
{\small\sf prigioniero@di.unimi.it}%
}
\date{}%
\maketitle\thispagestyle{empty}%
\maketitle\thispagestyle{empty}%
\begin{quotation}\small\noindent
  \textbf{Abstract.}\hspace{\labelsep}%
	A condition characterizing the class of regular languages which have several nonisomorphic minimal reversible automata is presented. The condition concerns the structure of the minimum automaton accepting the language under consideration.
	It is also observed that there exist reduced reversible automata which are not minimal, in the sense that all the automata obtained by merging some of their equivalent states are irreversible.
	Furthermore, a sufficient condition for the existence of infinitely many reduced reversible automata accepting a same language is given.
	It is also proved that, when the language is accepted by a unique minimal reversible automaton (that does not necessarily concide with the minimum deterministic automaton), then no other reduced reversible automata accepting it can exist.
\mbox{}\\
\textbf{Keywords:}\hspace{\labelsep}%
        reversible automata, minimal automata, regular languages
\end{quotation}

	\section{Introduction}
	\label{sec:intro}
A device is said to be \emph{reversible} when each configuration has exactly one predecessor and one successor, thus implying that there is no loss of information during the computation. On the other hand, as observed by Landauer, logical irreversibility is associated with physical irreversibility and implies a certain amount of heat generation~\cite{Lan61}.
In order to avoid such a power dissipation and, hence, to reduce the overall power consumption of computational devices, the possibility of realizing reversible machines looks appealing.

A lot of work has been done to study reversibility in different computational devices. Just to give a few examples in the case of general devices as Turing machines, Bennet proved that each machine can be simulated by a reversible one~\cite{Ben73}, while Lange, McKenzie, and Tapp proved that each deterministic machine can be simulated by a reversible machine which uses the same amount of space~\cite{LMT00}. As a corollary, in the case of a constant amount of space, this implies that each regular language is accepted by a \emph{reversible two-way deterministic finite automaton}. Actually, this result was already proved by Kondacs and Watrous~\cite{KondacsW97}. 

However, in the case of \emph{one-way} automata, the situation is different. In fact, as shown by Pin, the regular language~$a^*b^*$ cannot be accepted by any reversible automaton~\cite{Pin92}.\footnote{From now on, we will consider only \emph{one-way automata}. Hence we will omit to specify ``one-way'' all the times.} So the class of languages accepted by \emph{reversible automata} is a proper subclass of the class of regular languages.
Actually, there are some different notions of reversible automata in the literature.
In~1982, Angluin introduced reversible automata in algorithmic learning theory, considering devices having only one initial and only one final state~\cite{Ang82}. %
On the other hand, the devices considered in~\cite{Pin92}, besides a set of final states, can have multiple initial states, hence they can take a nondeterministic decision at the beginning of the computation. An extension which allows to consider nondeterministic transitions, without changing the class of accepted languages, has been considered by Lombardy~\cite{Lom02}, introducing and investigating \emph{quasi reversible automata}.
Classical automata, namely automata with a single initial state and a set of final states, have been considered in the works by Holzer, Jakobi, and Kutrib~\cite{Kut14,HJK15,Kut15}. In particular, in~\cite{HJK15} the authors gave a characterization of regular languages which are accepted by reversible automata. This characterization is given in terms of the structure of the minimum deterministic automaton. Furthermore, they provide an algorithm that, in the case the language is acceptable by a reversible automaton, allows to transform the minimum automaton into an equivalent reversible automaton, which in the worst case is exponentially larger than the given minimum automaton. In spite of that, the resulting automaton is minimal, namely there are no reversible automata accepting the same language with a smaller number of states. However, it is not necessarily unique, in fact there could exist different reversible automata with the same number of states accepting the same language.

\medskip

In this paper we continue the investigation of minimality in reversible automata.
Our first result is a condition that characterizes languages having several different minimal reversible automata. Even this condition is on the structure of the transition graph of the minimum automaton accepting the language under consideration: there exist different minimal reversible automata accepting a reversible language~$L$ if and only if there are transitions on different letters entering in a same state in the ``irreversible part'' of the minimum automaton for~$L$. We prove that such condition is always satisfied each time the ``irreversible part'' of the minimum automaton contains a loop.

We also observe that there exist reversible automata which are not minimal but they are reduced, in the sense that when we try to merge some of their equivalent states we always obtain an irreversible automaton.
Investigating this phenomenon more into details, we are able to find a language for which there exist arbitrarily large, and hence infinitely many reduced reversible automata. Indeed, we present a general construction that allows to obtain arbitrarily large reversible automata for all languages whose minimum deterministic automata satisfy a further structural condition besides the one given in~\cite{HJK15}. As a consequence of such condition, for each reversible language whose minimum deterministic automaton contains a loop in the ``irreversible part'', it is possible to construct infinitely many arbitrarily large reduced reversible automata. We know that our condition is not necessary and we leave as an open problem to find a characterization of the class of the languages having infinitely many reduced reversible automata.

Finally, in the last part of the paper we prove that when the minimal reversible automaton accepting a language is unique, there does not exist any other equivalent reduced reversible automaton. However, such unique minimal reversible automaton can be larger than the minimum automaton accepting the same language.

\section{Preliminaries}
\label{sec:prel}
In this section we recall some basic definitions and results useful in the paper. We assume the reader is familiar with standard notions from automata and formal language theory (see, e.g.,~\cite{HU79}).
Given a set~$S$, let us denote by~$\#S$ its cardinality and by~$2^S$ the family of all its subsets.
Given an alphabet~$\Sigma$, $|w|$~denotes the length of a string~$w\in\Sigma^*$ and~$\varepsilon$ the empty string. 

\medskip

A \emph{deterministic finite automaton} (\dfa\,for short) is a tuple~$A\!=\!(Q, \Sigma, \delta, q_I, F)$, where $Q$ is the finite set of \emph{states}, $\Sigma$ is the \emph{input alphabet}, $q_I\in Q$ is the \emph{initial state}, $F \subseteq Q$ is the set of \emph{accepting states}, and $\delta: Q\times\Sigma \rightarrow Q$ is the partial \emph{transition function}. The \emph{language accepted} by $A$ is $L(A) = \left\{ w \in \Sigma^* \mid \delta(q_I,w) \in F \right\}$.
 
The \emph{reverse} transition function of~$A$ is a function $\delta^R: Q\times\Sigma \rightarrow 2^Q$, with $\delta^R(p,a)= \{q \in Q \mid \delta(q,a)= p\}$. 

A state $p\in Q$ is \emph{useful} if~$p$ is \emph{reachable}, i.e., there is~$w \in \Sigma^*$ such that $\delta(q_I, w) = p$, and \emph{productive}, i.e., there is~$w \in \Sigma^*$  such that $\delta(p,w) \in F$. In this paper we only consider automata with all useful states.

We say that two states~$p,q \in Q$ are \emph{equivalent} if and only if for all $w\in \Sigma^*$, $\delta(p,w) \in F$ exactly when $\delta(q,w)\in F$.
When~$p\neq q$ are equivalent states, we can reduce the size of the automaton by ``merging''~$p$ and~$q$. This would imply to merge all the states reachable from~$p$ and~$q$ by reading a same string, namely the states~$\delta(p,w)$ and~$\delta(q,w)$, for~$w\in\Sigma^*$.

\medskip

Two automata $A$ and $A'$ are said to be \emph{equivalent} if they accept the same language, i.e., $L(A)=L(A')$.

\medskip

Given two \dfas~$A=(Q,\Sigma,\delta,q_I,F)$ and~$A'=(Q',\Sigma,\delta',q'_I,F')$, a \emph{morphism} $\varphi$ from~$A$ to~$A'$,
in symbols~$\varphi:A\rightarrow A'$, is a function~$\varphi:Q\rightarrow Q'$ such that:
\begin{itemize}
\item $\varphi(q_I)=q'_I$,
\item $\varphi(\delta(q,a))=\delta'(\varphi(q),a)$, and
\item $q\in F$ if and only if~$\varphi(q)\in F'$,
\end{itemize}
for~$q\in Q$, $a\in\Sigma$.

Notice that if there exists a morphism~$\varphi:A\rightarrow A'$ then it is \emph{unique}
and, for~$x,y\in\Sigma^*$, $\delta(q_I,x)=\delta(q_I,y)$ implies $\delta'(q'_I,x)=\delta'(q'_I,y)$, thus
implying that~$A$ and~$A'$ are equivalent.
We can observe that since in all automata we are considering all the states are useful, there exists
the morphism~$\varphi:A\rightarrow A'$ if and only if the automaton~$A'$ can be obtained from~$A$ after merging all pairs of states~$p,q$ of $A$, with~$\varphi(p)=\varphi(q)$
(and possibly renaming the states).
Hence~$\varphi^{-1}(s)$ is the set of states
of~$A$ which are merged in the state~$s$ of~$A'$, and the number of states of~$A'$ cannot exceed that of~$A$. 

We point out that the notion of morphism is strictly related to the Myhill-Nerode equivalence relation~\cite{Myh57,Ner58} in the following sense. If~$R_A$ and~$R_{A'}$ are the Myhill-Nerode relations associated with \dfas~$A$ and~$A'$, then~$R_{A'}$ is a refinement of~$R_A$ if and only if there is a morphism from~$A$ to~$A'$.

\medskip

\noindent
Let $\mathcal{C}$ be a family of \dfas\ and~$A\in\mathcal{C}$. We consider the following notions:
\begin{itemize}
	\item The automaton~$A$ is \emph{reduced} in~$\mathcal{C}$ if for each morphism~$\varphi: A\rightarrow A'$, either $\varphi$ is an isomorphism (i.e., $A$ and $A'$ are the same automaton up to a renaming of states) or the automaton $A'$ does not belong to~$\mathcal{C}$, that is, every automaton obtained from $A$ by merging some equivalent states does not belong to $\mathcal{C}$.

	\item The automaton~$A$ is \emph{minimal} in $\mathcal{C}$ if and only if each automaton in $\mathcal{C}$ has at least as many states as $A$.
	\item The automaton~$A$ is the \emph{minimum} in $\mathcal{C}$ if and only if, for each automaton~$A'$ in~$\mathcal{C}$, there exists a morphism $\varphi:A'\to A$.
	\end{itemize}
Notice that each minimal automaton in a family~$\mathcal{C}$ is reduced. Furthermore, if~$\mathcal{C}$ contains a minimum automaton~$M$, then~$M$ is also the only minimal and the only reduced automaton in~$\mathcal{C}$. This happens, for instance, when~$\mathcal{C}$ is the family of all \dfas\ accepting a given regular language~$L$. However, a family~$\mathcal{C}$ which does not have a minimum automaton, could contain reduced automata which are not minimal, as in some of the cases that will be presented in the paper.\footnote{In the conference version of the paper~\cite{LPP16}, we gave a weaker definition, by saying that~$A$ is minimum if it is the \emph{unique} minimal automaton in~$\mathcal{C}$. As we will show later in the paper (see Thm.~\ref{th:unique-minimal}), for the models we are considering the two notions are equivalent. However, in general, this is not true if we do not make any assumption about the family~$\mathcal{C}$. For example, in $\mathcal{C}= \{A \mid \text{$A$ is a \dfa\ s.t. $L(A)= a^*$ or $L(A)= (ab)^*$} \}$, the one state \dfa\ accepting $a^*$ is the only  minimal automaton, while the two state \dfa\ accepting $(ab)^*$ is reduced. Hence, the family~$\mathcal{C}$ does not have a minimum automaton.}

\medskip

A \emph{strongly connected component} (\scc)~$C$ of a \dfa\ $A = (Q, \Sigma, \delta, q_I, F)$ is a maximal subset of~$Q$ such that in the transition graph of~$A$ there exists a path between every pair of states in~$C$.
A \scc\ consisting of a single state~$q$, \emph{without} a looping transition, is said to be \emph{trivial}. Otherwise~$C$ is \emph{nontrivial} and, for each state~$q\in C$, there is a string~$w\in\Sigma^+$ such that~$\delta(q,w)=q$.

We consider a partial order $\preceq$ on the set of \sccs\ of $A$, such that, for two such components $C_1$ and $C_2$, $C_1\preceq C_2$ when either $C_1=C_2$ or no state in $C_1$ can be reached from a state in $C_2$, but a state in $C_2$ is reachable from a state in $C_1$. %
We write $C_1\not\preceq C_2$ when~$C_1\preceq C_2$ is false, namely, $C_1\neq C_2$ and either~$C_2\preceq C_1$ or~$C_1$ and~$C_2$ are incomparable.

\section{Reversible Automata}
\label{sec:reversible}

This section is devoted to recall basic notions and results related to reversible automata and to present some preliminary lemma that will be used in the paper.

\medskip

Given a \dfa~$A=(Q,\Sigma, \delta,q_I,F)$, a state $r \in Q$ is said to be \emph{irreversible} when~$\#\delta^R(r,a)>1$ for some~$a\in\Sigma$,
i.e., there are at least two transitions on the same letter entering~$r$, otherwise~$r$ is said to be \emph{reversible}.
The \dfa\ $A$~is said to be \emph{irreversible} if it contains at least one irreversible state, otherwise $A$ is \emph{reversible} (\revdfa\ for short).

According to the notion of reversible and irreversible states, each \dfa\ can be split in two parts:
the \emph{reversible part} and the \emph{irreversible part}. Roughly speaking, the irreversible part consists of all states that can be reached with a path which starts in an irreversible state, and of all transitions connecting those states. The reversible part consists of the remaining states and transitions, namely the states that can be reached from the initial state by visiting \emph{only} reversible states, and their outgoing transitions.
We notice that some of these transitions lead to irreversible states. The set of these transitions represents the ``border'' between the reversible and the irreversible part of the \dfa.
In Figure~\ref{fig:rev_irrev_parts} an example of this division is shown.
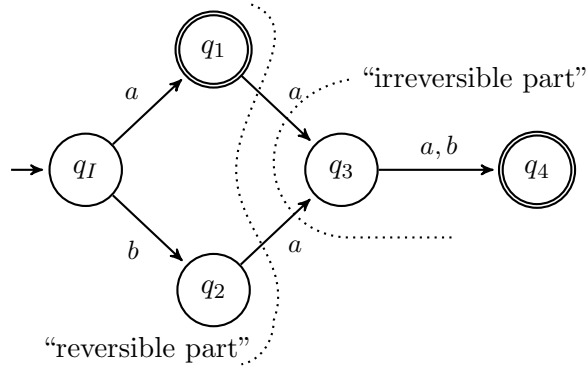
\begin{figure}[h]
	\centering
	\begin{tikzpicture}[->,>=stealth',shorten >=2pt,auto,node distance=2cm, thick,   main node/.style={circle,fill=black!20,draw,font=\sffamily}]
		\tikzstyle{every state}=[]

		\node[initial,initial text=,state] (q0) at (0,0) {$q_I$};
		\node[state,accepting] (q1) at (1.7,1.6) {$q_1$};
		\node[state] (q2) at (1.7,-1.6)	{$q_2$};

		\node[state] (q3) at (3.4,0) {$q_3$};
		\node[state,accepting] (q4) at (6,0) {$q_4$};

		\path[every node/.style={font=\sffamily\small}]
		(q0) edge  [] node[] {$a$} (q1)
		(q0) edge [below left] node[] {$b$} (q2)

		(q1) edge  [] node[] {$a$} (q3)
		(q2) edge [below right] node[] {$a$} (q3)

		(q3) edge [] node[] {$a,b$} (q4);

		\draw [-,dotted, thick] (2.2,2.2) to [out=-20,in=90] (2,0) to [out=-90,in=70] (2.5,-2) to [out=-90,in=10] node[left] {``reversible part''}  (2,-2.6);

		\draw [-,dotted, thick] (3.5,1.2)  node[right] {``irreversible part''}  to [out=190,in=90] (2.5,0.1) to [out=-90,in=180] (3.5,-0.9) to [out=-0,in=180] (5,-0.9); 

	\end{tikzpicture}
	\caption{
		A minimum \dfa\ with its reversible part (states $q_I$, $q_1$, $q_2$  and outgoing transitions) and its irreversible part (states $q_3$, $q_4$). The transitions from $q_1$ and $q_2$ on the symbol $a$ are at the border between the two parts
	}
	\label{fig:rev_irrev_parts}
\end{figure}

\medskip

As pointed out in~\cite{Kut15}, the notion of reversibility for a language is related to the computational model under consideration. In this paper we only consider \dfas.
Hence, by saying that a language~$L$ is \emph{reversible}, we refer to this model, namely we mean that there exists a \revdfa\ accepting~$L$.

The following result presents a characterization of reversible languages~\cite[Thm.~2]{HJK15}:
\begin{theorem}
\label{th:condition}
Let~$L$ be a regular language and~$M = (Q, \Sigma, \delta, q_I, F)$ be the minimum \dfa\ accepting~$L$. $L$ is accepted by a \revdfa\ if and only if there do not exist useful states $p, q \in Q$, a letter $a \in \Sigma$, and a string $w \in \Sigma^*$ such that $ p \neq q $, $\delta(p, a)= \delta(q, a)$, and $\delta(q, aw)= q$.
\end{theorem}
According to Theorem~\ref{th:condition}, a language~$L$ is reversible exactly when the minimum \dfa\ accepting it does not contain the ``forbidden pattern'' consisting of two transitions on a same letter~$a$ entering in a same state~$r$, with one of these transitions arriving from a state in the same \scc\ as~$r$. (See Figure~\ref{fig:forbidden_pattern}.)
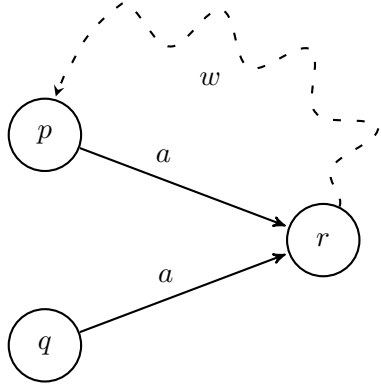
\begin{figure}%
	\centering
	\begin{tikzpicture}[->,>=stealth',shorten >=1pt,auto,thick]

		\node[state] (r) {$r$};
		\node[state] (p) [above left=0.7cm and 3cm of r] {$p$};
		\node[state] (q) [below left=0.7cm and 3cm of r] {$q$};

		\path[]
			(p) edge [pos=0.3125] node {$a$} (r)
			(q) edge [pos=0.5] node {$a$} (r)
			(r) edge [
				loosely dashed,
				bend right=95,
				looseness=1.1,
				-stealth,decorate,
				decoration={
					snake, 
					amplitude = 3mm,
					segment length = 12mm,
					post length=0.5mm
				}] node {$w$} (p);
	\end{tikzpicture}
	\caption{The ``forbidden pattern'': when it occurs in the transition graph, the minimum \dfa\ cannot be converted into a \revdfa.
		It is required that~$p \neq q$, but~$p$ or~$q$ could be equal to~$r$}
	\label{fig:forbidden_pattern}
\end{figure}

Notice that, since transitions entering in the initial state~$q_I$ can only arrive from states in the same \scc\ of~$q_I$,  if the language~$L$ is reversible, then the initial state~$q_I$ and all the states in the \scc\ containing it are reversible.

\medskip

An algorithm for converting a minimum irreversible \dfa\ $M$ into an equivalent minimal \revdfa\ $A$, when possible, has been obtained in~\cite{HJK15}. %
We present an outline of it.
At the beginning~$A$ is a copy of~$M$. Then, the algorithm considers (with respect to~$\preceq$) a minimal \scc\ $C_p$ that contains an irreversible state $p$ and replace it with~$\alpha$ of copies of itself, where~$\alpha$ is equal to the maximum number of transitions on a same letter incoming in a state of~$C_p$, i.e., $$ \alpha = \max\{\#\delta^R(\hat{p},a) \mid \hat{p}\in C_p, a \in \Sigma\}$$
and redistributes all incoming transitions among these copies, such that no state in the copies of $C_p$ is the target of two or more transitions on the same letter. Notice that all transitions that witness the irreversibility of states in $C_p$ come from outside of $C_p$, otherwise $M$ would have the forbidden pattern.
This process is iterated until all the states in~$A$ become reversible. (See  Algorithm~\ref{alg:dfa_to_revdfa}.)
As pointed out in~\cite{HJK15}, in some steps of the algorithm different choices are possible (see lines~\ref{alg:C_p} and~\ref{alg:red}). As a consequence, different minimal \dfas\ equivalent to~$M$ could be produced.

\begin{algorithm}[tb]
	\caption{Transformation from a minimum \dfa~$M$ to a minimal \revdfa.}
	\label{alg:dfa_to_revdfa}
	\begin{algorithmic}[1]
		\State Let~$A$ be a copy of~$M$
		\While {there are irreversible states in~$A$}
			\State $C_p \leftarrow$ a minimal \scc\ in~$A$ with at least an irreversible state $p$\label{alg:C_p}
			\State $\alpha \leftarrow \max \left\{ \#\delta^R(\hat p,a) \mid \hat p \in C_p, a\in\Sigma \right\}$\label{alg:alpha} 
			\State Replace~$C_p$ with $\alpha$ copies of itself
			\State Redistribute the transitions incoming to~$C_p$ between new copies\label{alg:red}
		\EndWhile
	\end{algorithmic}
\end{algorithm}

\medskip

We conclude this section by presenting some properties that will be used later in the paper to prove our results.

First of all, we observe that Algorithm~\ref{alg:dfa_to_revdfa}, applied to a whatever \dfa\ $M$ which does not contain the forbidden pattern, produces an equivalent \revdfa\ even when $M$ is not minimum.  For each \scc\ $C$ of $M$ which contains an irreversible state, the algorithm creates a number of copies which is equal to the maximum number of transitions entering in a state of $C$ with a same letter. The following property will be used in Section~\ref{sec:reduced} to prove Theorem~\ref{th:irrev}:

\begin{lemma}
\label{lemma:glez}
Let $M = (Q, \Sigma, \delta, q_I, F)$ be the minimum \dfa\ accepting a reversible language and $M'= (Q', \Sigma, \delta', q'_I, F')$ be an equivalent \dfa, with the morphism~$\varphi':M'\rightarrow M$, such that~$M'$ does not contain the forbidden pattern and its irreversible part is a copy of the corresponding part of~$M$, i.e., $\varphi'^{^{-1}}(\varphi'(p))=\{p\}$ for each state~$p$ in the irreversible part of~$M'$.
Given a \revdfa\ $A=(Q_A, \Sigma, \delta_A, q_{I_A}, F_A)$ obtained applying Algorithm~\ref{alg:dfa_to_revdfa} to $M'$,  with the morphism~$\varphi:A\rightarrow M$, 
if $p',p'' \in Q_A$ are two copies of $p\in Q'$ created by the algorithm, i.e., $\varphi(p')= \varphi(p'')= \varphi'(p)$, then there exists $z\in \Sigma^*$, $t',t''\in Q_A$, such that $\delta_A(t',z)= p'$, $\delta_A(t'',z)= p''$, $C_{\varphi(t')} \neq C_{\varphi'(p)}$ and $C_{\varphi(t'')} \neq C_{\varphi'(p)}$. 
\end{lemma}
\begin{proof}
During the step in which the component $C_p$ of~$M'$ is evaluated, the algorithm creates $\alpha$ copies of~$C_p$, where $\alpha$ is the maximum number of entering transitions in a same state $\hat{p}$ of $C_p$, with a same symbol $a\in \Sigma$ (Lines~\ref{alg:C_p} and~\ref{alg:alpha}). 
Since~$M'$ does not contain the forbidden pattern, these transitions come from a different \scc\ of~$M'$.
No further changes to the part corresponding to the component $C_p$ are done in the next steps of the algorithm.
Let $z' \in \Sigma^*$ be a string such that $\delta'(\hat{p},z')= p$ and $z=az'$. Then~$\delta_A(t',z)=p'$ and~$\delta_A(t'',z)=p''$, for two states~$t',t''$ that in~$A$ are in \sccs\ different from those containing~$p'$ and~$p''$. We are going to prove that in~$M$ the states~$\varphi(t')$ and~$\varphi(t'')$ are not in the same \scc\  which contains the state $\varphi(p')=\varphi(p'')=\varphi'(p)$.

If~$t'$ and~$t''$ are not equivalent, i.e., $\varphi(t')\neq\varphi(t'')$, then in~$M$ the states $\varphi(t')$ and~$\varphi(t'')$ cannot belong to the same \scc\ as~$\varphi'(p)$, otherwise~$M$ should contain the forbidden pattern.
If~$t'$ and~$t''$ are equivalent, suppose that the state~$t=\varphi(t')=\varphi(t'')$ belongs to the same \scc\ as~$\varphi'(p)$. Then it should exist a string~$w\in\Sigma^*$ such that~$\delta(\varphi'(p),w)=t$. Furthermore, because~$\delta(t,z)=\varphi'(p)$, there is a sequence of states $p_0,t_1,p_1,t_2,\ldots$ of~$M'$ such that $p_0=p$, $t_i=\delta'(p_{i-1},w)$, $p_i=\delta'(t_i,z)$ for $i>0$. (Notice that $\varphi(p_i)=\varphi'(p)$ and $\varphi(t_i)=t$.) At some point the sequence should contain a state which already appeared. It can be verified that this implies that~$M'$ contains the forbidden pattern, which is a contradiction.
Hence, we conclude that~$\varphi(t')$ and~$\varphi(t'')$ are not in the same \scc\ as~$\varphi'(p)$.
\end{proof}

A structural characterization of minimal \revdfas, which can be formulated as follows, has been obtained in~\cite[Thm.~6]{HJK15}:

\begin{theorem}
\label{th:char-minimal}
Let~$M=(Q,\Sigma,\delta,q_I,F)$ be a minimum \dfa\ and $A=(Q_A,\Sigma,\delta_A,q_{I_A},F_A)$ be a \revdfa\ equivalent to it, with the morphism $\varphi:A \to M$.
Then $A$ is minimal if and only if for every $q\in Q$, when the set $\varphi^{-1}(q)$ consisting of all states of $A$ which are equivalent to $q$ is not singleton, i.e., $\#\varphi^{-1}(q)>1$, there exists a word $x\in \Sigma^+$ such that for each $q' \in \varphi^{-1}(q)$, $\delta_A^R(q',x)$ is  defined and there are $q'_A,q''_A \in \varphi^{-1}(q)$ and $p'_A,p''_A \in Q_A$ such that $\delta_A(p'_A,x)=q'_A$, $\delta_A(p''_A,x)=q''_A$, and $p'_A$ and $p''_A$ are not equivalent, i.e., $\varphi(p'_A) \neq \varphi(p''_A)$. 
\end{theorem}

From Theorem~\ref{th:char-minimal} it follows that among all the states of a minimal \revdfa\ which are equivalent to a state in the irreversible part of the minimum automaton, there are two states which can be distinguished by ``moving back'' (i.e., considering the reverse transition function), with a same string~$x$.
We now prove that the same holds for each \revdfa. More precisely, we prove the following property, which will be used later in Section~\ref{sec:reduced} to prove Theorem~\ref{th:irrev}:

\begin{lemma}
\label{lemma:x}
Let~$M=(Q,\Sigma,\delta,q_I,F)$ be a minimum \dfa\ and let~$A=(Q_A,\Sigma,\delta_A,q_{I_A},F_A)$ be a \revdfa\ equivalent to it, with the morphism~$\varphi:A \to M$.
For each state~$q$ in the irreversible part of~$M$ there exist $x\in \Sigma^+$, $q'_A, q''_A \in \varphi^{-1}(q)$, and $p'_A,p''_A\in Q_A$ such that $\delta_A(p'_A,x)=q'_A$, $\delta_A(p''_A,x)=q''_A$, and $p'_A$ and $p''_A$ are not equivalent, i.e., $\varphi(p'_A) \neq \varphi(p''_A)$. 
\end{lemma}
\begin{proof}
Being~$q$ in the irreversible part, it should exist an irreversible state $q'\in Q$ and a string $z \in \Sigma^*$ such that $\delta(q',z)=q$. Furthermore, since $q'$ is irreversible, by definition there exists $a\in \Sigma$ such that $\#\delta^R(q',a)>1$. By taking $x=az$ and choosing two states $p',p''$ in $\delta^R(q',a)$, with~$p'\neq p''$, we obtain $\delta(p',x) = \delta(p'',x)=q$. 

Given~$p'_A,p''_A\in Q_A$ equivalent to $p',p''$, respectively, and considering $q'_A=\delta_A(p'_A,x)$, $q''_A=\delta_A(p''_A,x)$, we obtain~$\varphi(q'_A)=\varphi(q''_A)=q$. From $p'\neq p''$ it follows that $p'_A \neq p''_A$ and then $q'_A \neq q''_A$, otherwise $A$ cannot be reversible.
\end{proof}

\section{Minimal Reversible Automata}
\label{sec:minimal}
In~\cite{HJK15} it has been observed that there are reversible languages having several nonisomorphic minimal \revdfas.
In this section we deepen that investigation by presenting a characterization of the languages having a unique minimal
\revdfa. 

Let us start by presenting a series of preliminary results, that will be used later. Hence, from now on, let us fix a reversible language~$L$ 
and the minimum \dfa~$M=(Q,\Sigma,\delta,q_I,F)$ accepting it.

\begin{lemma}
\label{lemma:morph}
	Let~$A'\!=\!(Q',\Sigma,\delta',q'_I,F')$ be a \revdfa\ and~$A''\!=\!(Q'',\Sigma,\delta'',q''_I,F'')$ be a minimal \revdfa\ both accepting~$L$.
	Given the morphisms~$\varphi':A'\rightarrow M$ and~$\varphi'':A''\rightarrow M$, it holds that
	$\#\varphi'^{^{-1}}(s)\geq\#\varphi''^{^{-1}}(s)$, for each~$s\in Q$.
\end{lemma}
\begin{proof}
	By contradiction, suppose~$\#\varphi'^{^{-1}}(q)<\#\varphi''^{^{-1}}(q)$ for some state~$q$.

	Let us partition~$Q$ in the set~$Q_L=\{p\mid\exists w\in\Sigma^*~\delta(p,w)=q\}$ of the states from which~$q$
	is reachable and the set~$Q_R$ of remaining states.
	The sets~$Q'$ and~$Q''$ are partitioned in a similar way, by defining~$Q'_L=\varphi'^{^{-1}}(Q_L)$, $Q'_R=\varphi'^{^{-1}}(Q_R)$,
	$Q''_L=\varphi''^{^{-1}}(Q_L)$, $Q''_R=\varphi''^{^{-1}}(Q_R)$.

	First, let us suppose~$\#\varphi'^{^{-1}}(p)\leq\#\varphi''^{^{-1}}(p)$ for each~$p\in Q_L$.
	We build another automaton $A'''=(Q''',\Sigma,\delta''',q'''_I,F''')$, which starts the computation by simulating~$A'$ 
	using the states in~$Q'_L$ and, at some point, continues by simulating~$A''$ using the states in~$Q''_R$.
	In  particular:
	\begin{itemize}
	\item $Q'''=Q'_L\cup Q''_R$
	\item The transitions are defined as follows:
		\begin{itemize}
		\item For~$s\in Q''_R$, $a\in\Sigma$: $\delta'''(s,a)=\delta''(s,a)$.
		\item For~$s\in Q'_L$, $a\in\Sigma$, such that $\delta'(s,a)\in Q'_L$: $\delta'''(s,a)=\delta'(s,a)$.
		\item The remaining transitions, i.e., $\delta'''(s,a)$, in the case $s\in Q'_L$, $a\in\Sigma$, and $\delta'(s,a)\in Q'_R$,
		are obtained in the following way:

		Let us consider the set of states~$\{s_1,s_2,\ldots,s_k\}$ which are equivalent to~$s$ in~$A'$, i.e., $\varphi'(s_i)=\varphi'(s)$ for $i=1,\ldots,k$
		(notice that $s=s_h$ for some~$h\in\{1,\ldots,k\}$), and the set of states~$\{r_1,r_2,\ldots,r_j\}$ which are equivalent
		to~$s$ in~$A''$, i.e., $\varphi''(r_i)=\varphi'(s)$ for $i=1,\ldots,j$. Since~$j\geq k$ we can safely define~$\delta'''(s_i,a)=\delta''(r_i,a)$,
		for $i=1,\ldots,k$.
	  \end{itemize}
	\end{itemize}
	The resulting automaton~$A'''$ still recognizes the language~$L$, it is reversible and it has $\#Q'_L+\#Q''_R$ states.
	From~$\#\varphi'^{^{-1}}(p)\leq\#\varphi''^{^{-1}}(p)$, for each~$p\in Q_L$, and~$\#\varphi'^{^{-1}}(q)<\#\varphi''^{^{-1}}(q)$, it follows that~$\#Q'_L<\#Q''_L$, thus implying that the number of states of~$A'''$ is smaller than the one of~$A''$, which is a contradiction.

	In case~$\#\varphi'^{^{-1}}(p)>\#\varphi''^{^{-1}}(p)$ for some~$p\in Q_L$, we can apply the same construction, after switching the role of~$A'$ and~$A''$,
	so producing an equivalent~\revdfa~$\hat A'$ which is smaller than~$A'$ and still verifies~$\#\hat\varphi'^{^{-1}}(q)<\#\varphi''^{^{-1}}(q)$,
	for the morphism~$\hat\varphi':\hat A'\rightarrow M$. Then, we iterate the proof on the two \revdfas~$\hat A'$ and~$A''$.

	Hence, we can conclude that~$\#\varphi'^{^{-1}}(s)\geq\#\varphi''^{^{-1}}(s)$,  for each~$s\in Q$.
\end{proof}

Lemma~\ref{lemma:morph} allows to associate with each reversible language~$L$ and the minimum \dfa~$M=(Q,\Sigma,\delta,q_I,F)$ accepting it,
the function~$c:Q\rightarrow \IN^+$ such that, for $q\in Q$, $c(q)$ is the number of states equivalent to $q$ in any minimal \revdfa\ $A$ equivalent to $M$, i.e., $c(q)= \#\varphi^{-1}(q)$ for the morphism~$\varphi:A\rightarrow M$. Notice that~$c(q)=1$ if and only if~$q$ is in the reversible part of~$M$. As observed in Section~\ref{sec:reversible}, the initial state of~$M$ is always in the reversible part. Hence $c(q_I)=1$. Furthermore, each \revdfa\ accepting~$L$ should contain at least~$c(q)$ states equivalent to~$q$. These facts are summarized in the following result, where
we also show that~$c(q)$ has the same value for all states belonging to the same \scc\ of~$M$.

\begin{lemma}
\label{lemma:c}
  Let~$A$ be a \revdfa\ accepting~$L$, with the morphism~$\varphi:A\rightarrow M$.
  If two states $p,q$ of~$M$ belong to the same \scc\ of~$M$ then $\#\varphi^{-1}(p)=\#\varphi^{-1}(q)\geq c(p)$. Furthermore, if $A$ is minimal then $c(p)=c(q)=\#\varphi^{-1}(p)$.
\end{lemma}
\begin{proof}
Observe that since $p,q$ belong to the same \scc\ there exists $x\in \Sigma^*$ such that $\delta(q,x)=p$.
Let $\{q_1,q_2,\ldots,q_k\}=\varphi^{-1}(q)$ and $\{p_1,p_2,\ldots,p_j\}=\varphi^{-1}(p)$ be the sets of states in~$A$ which  are equivalent to $q$ and $p$, respectively.
We are going to prove that $k=j$.

For each $q_i$, there exists $p_{h_i}$ such that $\delta(q_i,x)=p_{h_i}$. Suppose $j<k$. In this case there are two indices~$i',i''$ such that~$p_{h_{i'}}=p_{h_{i''}}$ and then~$\delta(q_{i'},x)=\delta(q_{i''},x)=p_{h_{i'}}$, implying that the state $p_{h_{i'}}$ is irreversible, which is a contradiction. This means that~$j\geq k$. In the same way,
by interchanging the roles of $p$ and $q$, we can prove that~$k\geq j$, which leads to the conclusion~$j=k$. 

\smallskip

The facts that $\#\varphi^{-1}(p)\!\geq\!c(p)$ and, for~$A$ minimal, $\#\varphi^{-1}(p)\!=\!c(p)$, follow from Lemma~\ref{lemma:morph}.
\end{proof}

\noindent
In the following, for each \scc~$C$ of the transition graph of~$M$, we use~$c(C)$ to denote the value~$c(q)$, for~$q\in C$. Considering Algorithm~\ref{alg:dfa_to_revdfa}, we can observe that if~$C'$ is another \scc, then $C\preceq C'$ implies~$c(C)\leq c(C')$.

As a consequence of Lemma~\ref{lemma:c}, all the minimal \revdfas\ accepting~$L$ have the same ``state structure'', in the sense that they should contain exactly~$c(q)$ states equivalent to the state~$q$ of~$M$.
However, they could differ in the transitions (see Figure~\ref{fig:min} for an example).

\begin{figure}%
\centering
	\begin{minipage}{.33\textwidth}
		\centering
		\begin{tikzpicture}[->,>=stealth',shorten >=1pt,auto,node distance=2cm,thick]

			\node[accepting,initial,initial text={},state] (q0) {$q_I$};
			\node[state] (q1) [right  of=q0]	{$p$};
			\node[accepting,state] (q2) [below of=q0]	{$q$};

			\path[every node/.style={font=\sffamily\small}]
			(q0) edge  [bend left] node {$a$} (q1)
			(q1) edge  [bend left] node {$a$} (q0)

			(q0) edge [] node {$b$} (q2)
			(q1) edge [] node {$b$} (q2)

			(q2) edge [loop below] node {$a$} (q2);
		\end{tikzpicture}
	\end{minipage}%
	\begin{minipage}{.34\textwidth}
		\centering
		\begin{tikzpicture}[->,>=stealth',shorten >=1pt,auto,node distance=2cm,thick,main node/.style={circle,fill=blue!20,draw,font=\sffamily},main node2/.style={circle,fill=green!20,draw,font=\sffamily},main node3/.style={circle,fill=yellow!20,draw,font=\sffamily}]

			\node[accepting,initial,initial text={},state] (q0) {$q_I$};
			\node[state] (q1) [right  of=q0]	{$p$};
			\node[accepting,state] (q21) [below of=q0]	{$q'$};
			\node[accepting,state] (q22) [below of=q1]	{$q''$};

			\path[every node/.style={font=\sffamily\small}]
				(q0) edge  [bend left] node {$a$} (q1)
				(q1) edge  [bend left] node {$a$} (q0)

				(q0) edge [] node {$b$} (q21)
				(q1) edge [] node {$b$} (q22)

				(q21) edge [loop below] node {$a$} (q21)
				(q22) edge [loop below] node {$a$} (q22);
		 \end{tikzpicture}
	\end{minipage}%
	\begin{minipage}{.33\textwidth}
		\centering
		\begin{tikzpicture}[->,>=stealth',shorten >=1pt,auto,node distance=2cm,thick,main node/.style={circle,fill=blue!20,draw,font=\sffamily},main node2/.style={circle,fill=green!20,draw,font=\sffamily},main node3/.style={circle,fill=yellow!20,draw,font=\sffamily}]
			\node[accepting,initial,initial text={},state] (q0) {$q_I$};
			\node[state] (q1) [right  of=q0]	{$p$};
			\node[accepting,state] (s0) [below of=q0]	{$q'$};
			\node[accepting,state] (s1) [below of=q1]	{$q''$};

			\path[every node/.style={font=\sffamily\small}]
				(q0) edge  [bend left] node {$a$} (q1)
				(q1) edge  [bend left] node {$a$} (q0)

				(q0) edge [] node {$b$} (s0)
				(q1) edge [] node {$b$} (s1)

				(s0) edge [bend left] node {$a$} (s1)
				(s1) edge [bend left] node {$a$} (s0)

				(s0) edge [loop below,opacity=0] node {} (s0);
		\end{tikzpicture}
	\end{minipage}
	\caption{The minimum \dfa\ accepting the language $L=(aa)^*+a^*ba^*$, with two minimal nonisomorphic \revdfas}
	\label{fig:min}
\end{figure}
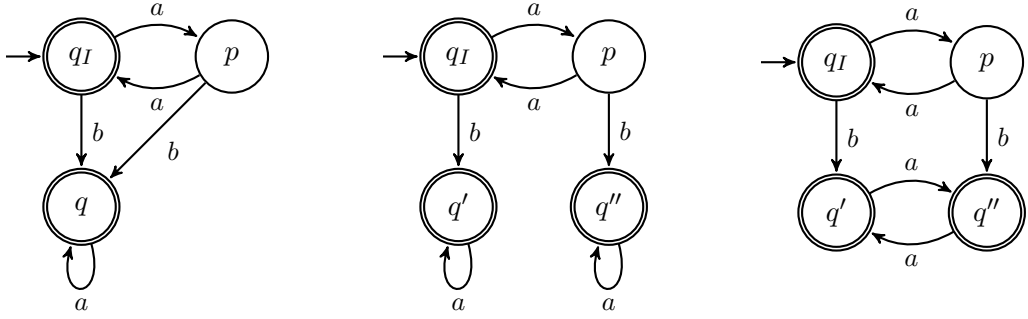

\begin{lemma}
\label{lemma:ab}
Let~$A'=(Q',\Sigma,\delta',q'_I,F')$ and~$A''=(Q'',\Sigma,\delta'',q''_I,F'')$ be two \revdfas\ accepting~$L$. 
Given the morphism~$\varphi'':A''\rightarrow M$, if there are no morphisms~$\varphi:A'\rightarrow A''$ then there exists a state $p\in Q$ with $\#\varphi''^{^{-1}}(p)>1$ such that either $p=q_I$, or
$$\delta^R(p,a) \neq \emptyset \mbox{ and } \delta^R(p,b) \neq \emptyset,$$\, 
for two symbols  $a,b\in \Sigma$, with $a\neq b$.
\end{lemma}
\begin{proof}
Since there are no morphisms~$\varphi:A'\rightarrow A''$, there exist $x,y\in \Sigma^*$ such that $\delta'(q'_I,x)=\delta'(q'_I,y)$ and $\delta''(q''_I,x)\neq \delta''(q''_I,y)$. 
Among all couples of strings with this property we choose one with $|xy|$ minimal.
Furthermore, we observe that it cannot be possible that~$x=y=\varepsilon$.

When~$x=\varepsilon$, we have~$\delta'(q'_I,\varepsilon)=\delta'(q'_I,y)=q'_I$ and, since~$M$ is minimum, $\delta(q_I,y)=q_I$. Hence,
$\varphi''(\delta''(q''_I,y))=\varphi''(q''_I)=q_I$. From $\delta''(q''_I,y)\neq q''_I=\delta''(q''_I,\varepsilon)$, we conclude that~$\#\varphi''^{^{-1}}(q_I)>1$.
The case~$y=\varepsilon$ is similar.

We now consider $x\neq\varepsilon$ and $y\neq\varepsilon$, i.e., $x=ua$, $y=vb$ for some $u,v\in \Sigma^*$ and $a,b\in \Sigma$.
Let $\delta'(q'_I,u)=q'$, $\delta'(q'_I,v)=r'$, $\delta'(q',a)=\delta'(r',b)=\bar{p}$, $\delta''(q''_I,u)=q''$, $\delta''(q''_I,v)=r''$,
$\delta''(q'',a)=s$, and  $\delta''(r'',b)=t$, for states~$q',r',\bar p\in Q'$, $q'',r'',s,t\in Q''$, with~$s\neq t$.

Suppose~$a=b$. Since $A'$ is reversible from $\delta'(q',a)=\delta'(r',a)=\bar{p}$ we get $q'=r'$. 
Furthermore $q''\neq r''$, otherwise $A''$ would be nondeterministic. 
Hence, on the strings~$u,v$ the automaton~$A'$ reaches the same state, while~$A''$ reaches two different states, against the minimality of~$|xy|$.  
Thus~$a\neq b$. 

Given the morphism~$\varphi':A'\rightarrow M$, let~$p=\varphi'(\bar p)$.
Since~$M$ is minimum, it turns out that~$\varphi''(s)=\varphi''(t)=\varphi'(\bar{p})=p$. From~$s\neq t$, we conclude that~$\#\varphi''^{^{-1}}(p)>1$. 
Furthermore, from the previous discussion, the reader can observe that there are transitions on symbols~$a$ and~$b$ entering in~$p$.
\end{proof}

We are now able to prove the following:

\begin{theorem}
	\label{th:minimal}
	Let $M=(Q,\Sigma,\delta,q_I,F)$ be the minimum \dfa\ accepting a reversible language $L$. The following statements are equivalent:
	\begin{enumerate}
		\item\label{1} There exists a state $p\in Q$ such that $c(p)>1$, $\delta^R(p,a) \neq \emptyset$, $\delta^R(p,b) \neq \emptyset$, for two symbols $a,b\in \Sigma$, with $a\neq b$. 

		\item\label{2} There exist at least two minimal nonisomorphic \revdfas\ accepting $L$.
	\end{enumerate}
\end{theorem}
\begin{proof}
(\ref{2}) implies (\ref{1}):
By Lemma~\ref{lemma:ab}, given two minimal nonisomorphic \revdfas~$A'$ and~$A''$ accepting~$L$, there is a state~$p$ such that $c(p)=\#\varphi''^{^{-1}}(p)>1$.
Furthermore, since~$c(q_I)=1$, $p\neq q_I$. Hence, $\delta^R(p,a) \neq \emptyset$, $\delta^R(p,b) \neq \emptyset$, for two symbols  $a,b\in \Sigma$, with $a\neq b$.

\medskip

(\ref{1}) implies (\ref{2}):
Let $w \in \Sigma^*$ be a string of minimal length such that $\delta(q_I,w)=p$,  $a\in\Sigma$ be its last symbol, i.e., $w=xa$, with~$x\in \Sigma^*$.
Let~$b\in \Sigma$ be a symbol with  $b \neq a$ and $\delta^R(p,b) \neq \emptyset$.
Given a minimal \revdfa~$A'=(Q',\Sigma,\delta',q'_I,F')$ accepting~$L$ and the morphism~$\varphi':A'\rightarrow M$, we consider the state~$\hat{p}=\delta'(q'_I,w)$. Then~$\varphi'(\hat p)=p$.

We show how to build a minimal \revdfa~$A''$ nonisomorphic to~$A'$. The idea is to use the set of states~$Q'$ as in~$A'$ and to modify only the transitions which simulates the transitions that in~$M$ enter the state~$p$ with the letter~$b$. There are different cases.

When~$\delta'^R(\hat{p},b)= \emptyset$, it should exist~$\tilde{p} \in \varphi'^{^{-1}}(p)$ such that~$\tilde{p}\neq\hat p$ and $\delta'(\tilde{q},b)= \tilde{p}$,
for some~$\tilde{q} \in Q'$. 
The automaton~$A''$ is defined as~$A'$, with the only difference that the transition $\delta'(\tilde{q},b)= \tilde{p}$ is replaced by $\delta''(\tilde{q},b)= \hat{p}$. 
To prove that it is nonisomorphic to $A'$, we consider a string~$y \in \Sigma^*$ of minimal length such that $\delta'(q'_I,y) = \tilde{q}$.
Then~$\delta'(q'_I,yb)=\tilde{p} \neq \delta'(q'_I,w)=\hat{p}$, while $\delta''(q'_I,yb)=\hat{p} = \delta''(q'_I,w)$.

When~$\delta'^R(\hat{p},b) \neq \emptyset$ we can use one of the following possibilities:
\begin{itemize}
\item If there exists $\tilde{p} \neq \hat{p}$ such that $\delta'^R(\tilde{p},b) \neq \emptyset$, then it should also exist~$\tilde{q},\hat{q} \in Q'$ with $\tilde{q} \neq \hat{q}$ such that $\delta'(\tilde{q},b) = \tilde{p}$ and  $\delta'(\hat{q},b) = \hat{p}$.
The automaton $A''$ is defined by switching the destinations of these two transitions, namely by replacing them by $\delta''(\tilde{q},b)= \hat{p}$ and $\delta''(\hat{q},b)= \tilde{p}$.
The proof that $A'$ and $A''$ are non isomorphic is exactly the same as in the previous case.

\item If there exists $\tilde{p} \neq \hat{p}$ such that $\delta'^R(\tilde{p},b) = \emptyset$, then we can consider $\hat{q}$ such that $\delta'(\hat{q},b)= \hat{p}$, and
define~$A''$ by replacing this transition by $\delta''(\hat{q},b)= \tilde{p}$. 
Let $y \in \Sigma^*$ be a string of minimal length such that $\delta'(q'_I,y) = \hat{q}$. Then~$\delta'(q'_I,yb)=\hat{p} = \delta'(q'_I,w)$. On the other hand $\delta''(q'_I,yb)=\tilde{p} \neq \hat{p} =\delta''(q'_I,w)$. Hence, $A'$ and $A''$ are nonisomorphic. 
\end{itemize}
Finally, we observe that the construction preserves the reversibility and, in both cases, the automaton~$A''$ has the same number of states as~$A'$. Hence, $A''$ is minimal.
\end{proof}

In other terms, Theorem~\ref{th:minimal} gives a characterization of reversible languages having a unique minimal \revdfa\ as the languages such that in the their minimum \dfas\ all transitions entering in each state in the irreversible parts are on the same symbol.

When the minimum \dfa\ accepting a reversible language contains a loop in the irreversible part, the previous condition is
always false, hence there exist at least two minimal nonisomorphic \revdfas. This is proved in the following result:

\begin{theorem}
\label{th:loop}
Let $M=(Q,\Sigma,\delta,q_I,F)$ be the minimum \dfa\ accepting a reversible language $L$. If there exists an 
irreversible state $q\in Q$ such that the language accepted
by computations starting in~$q$ is infinite, then there exists a state $p\in Q$ such that $c(p)>1$, 
$\delta^R(p,a) \neq \emptyset$ and $\delta^R(p,b) \neq \emptyset$,
for two symbols  $a,b\in \Sigma$, with $a\neq b$.
\end{theorem}
\begin{proof}
Let $p\in Q$ be a state reachable from $q$ which belongs to a nontrivial \scc~$C$. Hence~$c(p)>1$. Among all possibilities, we choose~$p$ in such a way that all the other states on a fixed path
from~$q$ to~$p$ does not belong to~$C$. Since~$C$ is nontrivial, it should exist a transition from a state of~$C$, which enters in~$p$.
Let~$a\in\Sigma$ be the symbol of such transition. Furthermore, it should exist another transition which enters in~$p$ from a state which does not belong to~$C$. (If~$p\neq q$ then we can take the last transition on the fixed path. Otherwise, since the initial state is always reversible, we have~$q\neq q_I$, and so we can take the last transition entering in~$q$ on a path from~$q_I$.) Let~$b$ the symbol of such transition. If~$a=b$ the automaton~$M$ would contain the forbidden pattern (cfr.\ Theorem~\ref{th:condition}), thus implying that~$L$ is not reversible.
Hence, we conclude~$a\neq b$.
\end{proof}

As a consequence of Theorems~\ref{th:minimal} and~\ref{th:loop}, we obtain the following:

\begin{corollary}
\label{cor:unique-minimal}
	Let~$L$ be a reversible language accepted by a unique minimal \revdfa, and let~$M=(Q,\Sigma,\delta,q_I,F)$ be the minimum \dfa\ for~$L$.
	Then all the loops in~$M$ are in the reversible part. Furthermore, for each state~$q$ in the irreversible part of~$M$, all the transition incoming to~$q$ are on a same letter, i.e., $\#\{a\in\Sigma\mid\delta^R(q,a)\neq\emptyset\}=1$.
\end{corollary}

We point out that there are reversible languages whose minimum \dfa\ does not contain any loop in the irreversible part, which does not have a unique minimal \revdfa.
Indeed, in~\cite{HJK15} an example with a finite language is presented.

\section{Reduced Reversible Automata}
\label{sec:reduced}
In this section we show that there exist \revdfas\ which are reduced but not minimal, namely they have more states than equivalent minimal \revdfas, but merging some of their equivalent states would produce an irreversible automaton.
Furthermore, we will prove that there exist reversible languages having arbitrarily large reduced \revdfas\ and, hence, infinitely many reduced \revdfas.

In Figure~\ref{fig:red} a reduced \revdfa\ equivalent to the \dfas\ in Figure~\ref{fig:min} is depicted. If we try to merge two states in
the loop, then the loop  collapses to unique state, so producing the minimum \dfa, which is irreversible.
Actually, this example can be modified by using a loop of~$N$ states: if (and only if)~$N$ is prime, we get a reduced
automaton. This is a special case of the construction which we are now going to present:

\begin{figure}%
	\centering
	\begin{tikzpicture}[->,>=stealth',shorten >=1pt,auto,node distance=2cm,thick,main node/.style={circle,fill=blue!20,draw,font=\sffamily},main node2/.style={circle,fill=green!20,draw,font=\sffamily},main node3/.style={circle,fill=yellow!20,draw,font=\sffamily}]
		\node[accepting,initial,initial text={},state] (q0) {$q_I$};
		\node[state] (q1) [right  of=q0]	{$p$};
		\node[accepting,state] (s0) [below of=q0]	{$q_0$};
		\node[accepting,state] (s1) [below of=q1]	{$q_1$};
		\node[accepting,state] (s2) [] at (3,-3.6)	{$q_2$};
		\node[accepting,state] (s3) [] at (1,-4.5)	{$q_3$};
		\node[accepting,state] (s4) [] at (-1,-3.6)	{$q_4$};

		\path[every node/.style={font=\sffamily\small}]
			(q0) edge  [bend left] node {$a$} (q1)
			(q1) edge  [bend left] node {$a$} (q0)

			(q0) edge [] node {$b$} (s0)
			(q1) edge [] node {$b$} (s1)

			(s0) edge  node {$a$} (s1)
			(s1) edge [bend left] node {$a$} (s2)
			(s2) edge [bend left] node {$a$} (s3)
			(s3) edge [bend left] node {$a$} (s4)
			(s4) edge [bend left] node {$a$} (s0);
	\end{tikzpicture}
	\caption{A reduced \revdfa}
	\label{fig:red}
\end{figure}

\noindent

\begin{theorem}
\label{th:irrev}
	Let~$M=(Q,\Sigma,\delta,q_I,F)$ be the minimum \dfa\ accepting a reversible language~$L$.
	Suppose that~$M$ contains a loop and a path from a state in the loop to some state~$s\in Q$ on a nonempty string which ends by the symbol~$a\in\Sigma$.
	If there is a symbol~$b\in\Sigma$, with~$b\neq a$, $\delta^R(s,b)\neq\emptyset$, and either~$\#\delta^R(s,b)>1$ or there is~$r\in Q$ such that $\delta(r,b)=s$ and~$c(r)>1$, then there exist infinitely many nonisomorphic reduced \revdfas\ accepting~$L$.
\end{theorem}
\begin{proof}
	First of all, we notice that the hypothesis on the transitions entering in~$s$ implies that~$c(s)>1$.
	Let $\ell$ be the loop satisfying the condition in the statement. 
	Let us fix a state~$q$ in $\ell$. 

	Let~$A$ be a minimal \revdfa\ accepting~$L$, obtained applying Algorithm~\ref{alg:dfa_to_revdfa} and~$N\geq c(q)$ an integer.
The idea is to modify~$A$ by replacing the part corresponding to the \scc~$C_q$ containing  $q$, with~$N$ copies of each state in~$C_q$
	and arranging the transitions in such a way that all the states in these~$N$ copies form one \scc, without changing the
	accepted language. Furthermore, all \sccs\ that follow~$C_q$ will be replicated a certain number of times. More precisely, we build a \dfa\ $A_N=(Q_N,\Sigma,\delta_N,q_{I_N},F_N)$ using the following steps:
	\begin{enumerate}
		\item[(i)]
			We put in~$A_N$ all the states of~$A$ which correspond to \sccs~$C$ of~$M$ with~$C_q\not\preceq C$ and all the transitions between these states.
		\item[(ii)]
			We add $N$ copies of the states in~$C_q$ to the set of states of~$A_N$. Given a state~$r\in C_q$, let us denote its copies as~$r_0,r_1,\ldots,r_{N-1}$.
		\item[(iii)]
			We fix a transition~$\delta(q,\sigma)=q'$ of~$M$, with~$q,q'\in C_q$ and $\sigma \in \Sigma$. For $i=0,\ldots,N-1$, 
			we define~$\delta_N(q_i,\sigma)=q'_{(i+1)\bmod N}$, and for the remaining transitions, namely $\delta(r,\gamma)=r'$ with~$(r,\gamma)\neq(q,\sigma)$, 
			we define~$\delta_N(r_i,\gamma)=r'_i$.
			In this way in~$A_N$ we have~$N$ copies of the \scc~$C_q$, modified in such a way that the transition from~$q_i$ on~$\sigma$
			in copy~$i$ leads to the state~$q'_{(i+1)\bmod N}$ in copy~$(i+1)\bmod N$.
		\item[(iv)]
			We add to~$A_N$ each transition that in~$A$ leads from a state added in~(i) to one state in the first $c(q)$ copies of~$C_q$ added in~(ii).
		\item[(v)]
			We complete the construction by adding to the part of automaton obtained in steps (i-iv) one copy of each remaining \scc~$C$ of~$M$, namely \sccs\ $C$ such that $C_q \preceq C$ with $C_q \neq C$, and suitable transitions from the already constructed part, in order to accept~$L$. We apply Algorithm~\ref{alg:dfa_to_revdfa} to the \dfa~$M$, in order to derive a \revdfa. 	  
	\end{enumerate}
	Let~$A_N$ be the \revdfa\ accepting~$L$ so obtained. 
	Considering the hypothesis of the theorem, we now prove that~$A_N$ contains two states~$s'$ and~$s''$ equivalent to~$s$ and two states~$r'$ and~$r''$ such that~$\delta_N(r',b)=s'$, $\delta_N(r'',b)=s''$ and the states~$\delta(r',x)$ and~$\delta(r'',x)$ are not equivalent, for a string~$x\in\Sigma^*$. %
	\begin{itemize}
		\item
			If~$\#\delta^R(s,b)>1$ then~$M$ should contain two inequivalent states such that, reading the symbol~$b$, the state~$s$ is reached.
			Hence, in~$A_N$ there are~$r'$ and~$r''$ which are not equivalent, while the states~$\delta_N(r',b)$ and~$\delta(r'',b)$ are equivalent to~$s$.
			In this case we take~$x=\varepsilon$.
		\item
			Otherwise, there is~$r\in Q$ such that~$\delta(r,b)=s$ and~$c(r)>1$.
			By Lemma~\ref{lemma:x}, among all the states equivalent to~$r$ in~$A_N$, it should exist~$r'$ and~$r''$ such that~$\delta_N^R(r',x)$ and~$\delta_N^R(r'',x)$ are not equivalent for a string~$x\in\Sigma^+$.
	\end{itemize}
	In both cases we choose~$s'=\delta_N(r',b)$ and~$s''=\delta_N(r'',b)$.

	Moreover, we observe that from the existence of a path from the loop $\ell$ (and hence from~$q$) to~$s$, the automaton~$A_N$ contains at least~$N$ states~$s_0,s_1,\ldots,s_{N-1}$ that are copies of~$s$ (not necessarily different from ~$s',s''$), each one of them having an entering transition on the letter~$a$. More precisely, $\delta_N(q_i,u)=s_i$, for $i=0,\ldots,N-1$, where~$u$ is the string labeling the path from~$q$ to~$s$.
	\begin{enumerate}
		\item[(vi)]
			We modify~$A_N$ in such a way that~$s',s''$ belong to $\{s_0,s_1,\ldots,s_{N-1}\}$. (E.g., if~$s'\notin\{s_0,s_1,\ldots,s_{N-1}\}$ then we can move the transition from~$r'$ on~$b$ to a state~$s_i$ which does not have any entering transition on~$b$, if any; if $s_i$ already has an entering transition on~$b$, we can move it to~$s'$, that is, making a switch --- this transformation preserves reversibility and the accepted language --- then $s_i$ becomes the ``new'' $s'$.)
	\end{enumerate}
	We are now going to prove that if~$N$ is prime then the automaton~$A_N$ finally obtained by this process is reduced. 
	To this aim we shall prove that if we try to merge two equivalent states~$p',p''$ ($p'\neq p''$) of~$A_N$ then we obtain an irreversible automaton.
	The proof is divided in three cases:
	\begin{enumerate}
		\item[(1)] 
			$p',p''$ are equivalent to a state~$p$ of~$M$ with~$C_q\not\preceq C_p$, where~$C_p$ denotes the \scc\ containing~$p$.\\
			These states have been added at step (i), copying them from the minimal \revdfa~$A$. By Lemma~\ref{lemma:c}, $A$ contains exactly~$c(p)$ states equivalent to~$p$.
			Hence, merging~$p'$ and~$p''$, the resulting automaton would contain less than~$c(p)$ states equivalent to~$p$ and, hence, it cannot be reversible.
		\item[(2)]
			$p',p''$ are equivalent to a state~$p$ of~$M$ belonging to~$C_q$.\\
			First, suppose~$p'=q_0$ and~$p''=q_j$, $0<j<N$.
			Considering step (iii), we observe that there is a string~$z$ such that~$\delta(q',z)=q$, then~$\delta(q,w)=q$ and~$\delta_N(q_i,w)=q_{(i+1)\bmod N}$, where~$w=\sigma z$.
			Thus, for each~$k\geq 0$, $\delta(q_0,w^{k(N-j)})=q_{k(N-j)\bmod N}$ and $\delta(q_j,w^{k(N-j)})=q_{j+k(N-j)\bmod N}=q_{(k-1)(N-j)\bmod N}$.
			Hence, merging~$q_0$ and~$q_j$ would imply merging all the states whose indices are in the set~$\{k(N-j)\bmod N\mid k\geq 0\}$, which,
			being~$N$ prime, coincides with~$\{0,\ldots,N-1\}$. As a consequence, all states~$q_i$ should collapse in a unique state. 
			Since considering step~(vi) $\delta(q_i,u)=s_i$, this would also imply that~$s'$ and~$s''$ collapses. 
			However,  since the states~$\delta_N^R(s',xb)=\delta_N^R(r',x)$ 
			and~$\delta_N^R(s'',xb)=\delta_N^R(r'',x)$ are not equivalent, this would introduce some irreversible state.

			If~$p'\neq q_0$, then we can always find a string~$y$ such that~$\delta_N(p',y)=q_0$. Using the transitions introduced at step~(iii), we get that~$\delta_N(p'',y)=q_j$, for some~$0<j<N$.
			Hence, merging~$p'$ and~$p''$ would imply merging~$q_0$ and~$q_j$, so reducing to the previous case.
		\item[(3)]
			$p',p''$ are equivalent to a state~$p$ of~$M$, such that~$C_q\neq C_p$ and~$C_q\preceq C_p$.\\
			Let~$\varphi:A_N\rightarrow M$ be the morphism from $A_N$ to $M$. Hence, $\varphi(p') = \varphi(p'')=p$.

			\begin{enumerate}
				\item[(a)]
					If $C_p=C_{q_I}$ then some copies of the component $q_I$ have been made.
					Since, as observed in Section~\ref{sec:reversible}, the initial state is always reversible ($c(q_I)=1$) and~$C_{q_I}$ is minimum with respect to~$\preceq$, the only way to obtain copies of~$C_{q_I}$ is to add them in step~(ii), i.e., $C_q=C_{q_I}$, thus implying~$C_q=C_p$ which gives a contradiction.

				\item[(b)]
					If $C_p \neq C_{q_I}$, then by Lemma~\ref{lemma:glez} there exist states $t',t''\in Q_N$ such that $\delta_N(t',z)=p'$, $\delta_N(t'',z)= p''$ for some $z\in \Sigma^*$ where $C_{\varphi(t')} \neq C_p$ and $C_{\varphi(t'')} \neq C_p$.
					Now we analyze some cases:
					\begin{itemize}
						\item
							If $t',t''$ are not equivalent then we cannot merge $p'$ and $p''$.
							In fact, it would imply to introduce some irreversible states.
						\item
							If $t',t''$ are equivalent and hence $\varphi(t') = \varphi(t'')= t \in Q$  then we have the following possibilities:
							\begin{enumerate}
								\item[(I)]
									$C_t = C_q$: this brings us to the case~(2).
								\item[(II)]
									$C_q \not\preceq C_t$: this brings us to the case~(1). 
								\item[(III)]
									$C_q \preceq C_t$ (and $C_q \neq C_t$): we can iterate case~(3) we are considering, by substituting $p'$ and $p''$ with $t'$ and $t''$, respectively. 
									At some time, we will finally stop either because $t'$ and $t''$ are not equivalent or because we reach one of the cases~(I) or~(II), or because $t'$ and $t''$ are equivalent to a state in the component $C_{q_I}$, so obtaining case~(a). 
							\end{enumerate}
					\end{itemize}
			\end{enumerate}
	\end{enumerate}
	In summary, for each prime number~$N\geq c(q)$ we obtained a reduced \revdfa~$A_N$ with more than~$N$ states 
	accepting the language~$L$.
  	Hence, we can conclude that there exist infinitely many nonisomorphic reduced \revdfa\ accepting~$L$.
\end{proof}

\begin{figure}%
	\centering
	\begin{minipage}{.5\textwidth}
		\centering
		\begin{tikzpicture}[->,>=stealth',shorten >=1pt,auto,node distance=2cm,thick,main node/.style={circle,fill=blue!20,draw,font=\sffamily}]

			\node[initial,initial text=,state] (q0) {$q_I$};
			\node[state] (q1) [below left of=q0]	{$q$};
			\node[state] (q2) [below right of=q0]	{};

			\node[state,accepting] (q4) [below right of=q2]	{};
			\node[state,accepting] (q5) [below left of=q4]	{};

			\node[accepting,state] (q6) [below of=q5]	{$s$};

			\path[every node/.style={font=\sffamily\small}]
				(q0) edge  [] node[left=4pt] {$a$} (q1)
				(q0) edge  [] node[right=4pt] {$c$} (q2)

				(q2) edge  [] node[right=4pt] {$c$} (q4)

				(q4) edge  [] node[left=4pt] {$a$} (q5)
				(q4) edge  [bend left] node[left=4pt] {$b$} (q6)

				(q5) edge  [] node {$b$} (q6)

				(q1) edge [loop left] node {$c$} (q1)
				(q1) edge [bend right] node[left=4pt] {$a$} (q6)
			;
		\end{tikzpicture}
	\end{minipage}%
	\begin{minipage}{.5\textwidth}
		\centering
		\begin{tikzpicture}[->,>=stealth',shorten >=1pt,auto,node distance=2cm,thick,main node/.style={circle,fill=blue!20,draw,font=\sffamily}]
			\node[initial,initial text=,state] (q0) {$q_I$};
			\node[state] (q1) [below left of=q0]	{};
			\node[state] (q2) [below right of=q0]	{};

			\node[state,accepting] (q4) [below right of=q2]	{};

			\node[state,accepting] (q5) [below left of=q4]	{};

			\node[accepting,state] (q6) [below of=q5]	{};
			\node[accepting,state] (q7) at (q4 |- q6)	{};

			\path[every node/.style={font=\sffamily\small}]
				(q0) edge  [] node[left=4pt] {$a$} (q1)
				(q0) edge  [] node[right=4pt] {$c$} (q2)

				(q2) edge  [] node[right=4pt] {$c$} (q4)
				(q4) edge  [] node[left=4pt] {$a$} (q5)

				(q4) edge  [] node[left=4pt] {$b$} (q7)
				(q5) edge  [] node {$b$} (q6)

				(q1) edge [loop left] node {$c$} (q1)
				(q1) edge [bend right] node[left=4pt] {$a$} (q6)
			;
		\end{tikzpicture}
	\end{minipage}
	\begin{tikzpicture}[->,>=stealth',shorten >=1pt,auto,node distance=2cm,thick,main node/.style={circle,fill=blue!20,draw,font=\sffamily}]

		\node[initial,initial text=,state] (q0) {$q_I$};
		\node[state] (p0) [below left of=q0]	{$q_0$};
		\node[state] (p1) [below left of=p0] 	{$q_1$};
		\node[state] (p2) [below right of=p0] 	{$q_2$};

		\node[state] (q2) [below right of=q0]	{};

		\node[state,accepting] (q4) [below right of=q2]	{};

		\node[state,accepting] (q5) [below left of=q4]	{};
		\node[accepting,state] (q6) [below of=q5]	{};
		\node[accepting,state] (q7) at (q4 |- q6)	{};
		\node[accepting,state] (q8) at (p1 |- q6)	{};

		\path[every node/.style={font=\sffamily\small}]
			(q0) edge  [] node[left=4pt] {$a$} (p0)
			(q0) edge  [] node[right=4pt] {$c$} (q2)

			(q2) edge  [] node[right=4pt] {$c$} (q4)
			(q4) edge  [] node[left=4pt] {$a$} (q5)

			(q4) edge  [] node[left=4pt] {$b$} (q7)
			(q5) edge  [] node {$b$} (q6)

			(p0) edge [bend right] node {$c$} (p1)
			(p1) edge [bend right] node {$c$} (p2)
			(p2) edge [bend right] node {$c$} (p0)

			(p2) edge [bend right] node [left=4pt] {$a$} (q6)
		;
		\draw[->] (p1) edge [out=270,in=180+45, looseness=0.8]  node[] {$a$} (q7);
		\draw[->] (p0) edge [out=180,in=180, looseness=0.9]  node[] {$a$} (q8);

		\end{tikzpicture}
	\vskip 1em
\caption{The minimum \dfa\ accepting a reversible language $L = ac^*a+cc(\varepsilon+a)(\varepsilon+b)$, with an equivalent  minimal \revdfa\ and an equivalent reduced \revdfa}
\label{fig:min_dfa1}
\end{figure}
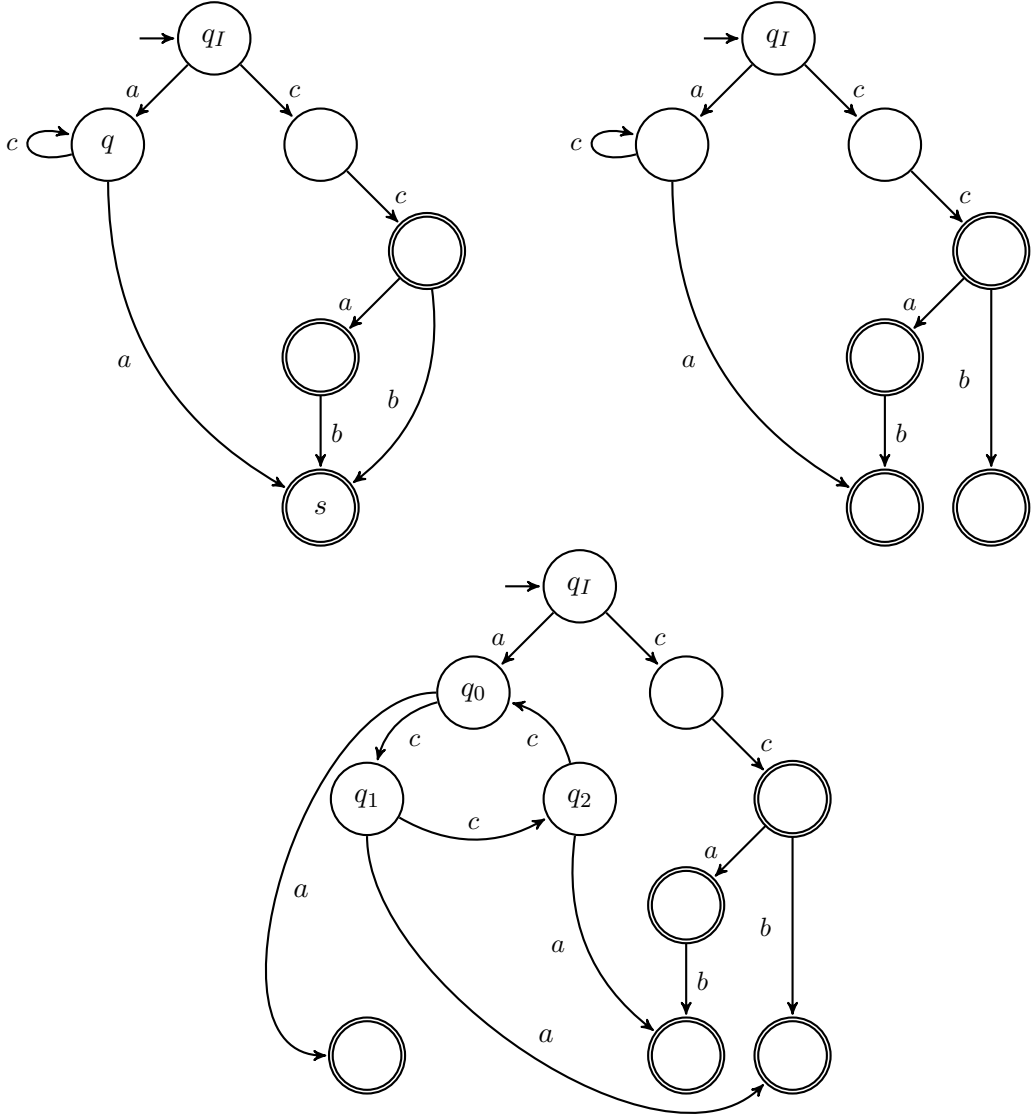

\begin{figure}%
	\centering
	\begin{minipage}{.5\textwidth}
		\centering
		\begin{tikzpicture}[->,>=stealth',shorten >=1pt,auto,node distance=2cm,thick,main node/.style={circle,fill=blue!20,draw,font=\sffamily}]

			\node[initial,initial text=,state] (q0) {$q_I$};
			\node[state] (q1) [below left of=q0]	{$q$};
			\node[state] (q2) [below right of=q0]	{};

			\node[state] (q3) [below left of=q2]	{};
			\node[state,accepting] (q4) [below right of=q2]	{};

			\node[state,accepting] (q5) [below left of=q4]	{$r$};

			\node[accepting,state] (q6) [below of=q5]	{$s$};

			\path[every node/.style={font=\sffamily\small}]
				(q0) edge  [anchor=south] node {$a$} (q1)
				(q0) edge  [anchor=south] node {$c$} (q2)

				(q2) edge  [anchor=south] node {$b$} (q3)
				(q2) edge  [anchor=south] node {$c$} (q4)

				(q3) edge  [anchor=south] node {$a$} (q5)
				(q4) edge  [anchor=south] node {$b$} (q5)

				(q5) edge  [] node {$a,b$} (q6)

				(q1) edge [loop left] node {$c$} (q1)
				(q1) edge [bend right] node[left=4pt] {$a$} (q5)
			;
		\end{tikzpicture}
	\end{minipage}%
	\begin{minipage}{.5\textwidth}
		\centering
		\begin{tikzpicture}[->,>=stealth',shorten >=1pt,auto,node distance=2cm,thick,main node/.style={circle,fill=blue!20,draw,font=\sffamily}]
			\node[initial,initial text=,state] (q0) {$q_I$};
			\node[state] (q1) [below left of=q0]	{$q$};
			\node[state] (q2) [below right of=q0]	{};

			\node[state] (q3) [below left of=q2]	{};
			\node[state,accepting] (q4) [below right of=q2]	{};

			\node[state,accepting] (q5_1) [below right of=q3]	{};
			\node[state,accepting] (q5_2) [left of=q5_1]	{};

			\node[accepting,state] (q6_1) [below of=q5_1]	{};
			\node[accepting,state] (q6_2) [below of=q5_2]	{};

			\path[every node/.style={font=\sffamily\small}]
				(q0) edge  [anchor=south] node {$a$} (q1)
				(q0) edge  [anchor=south] node {$c$} (q2)

				(q2) edge  [anchor=south] node {$b$} (q3)
				(q2) edge  [anchor=south] node {$c$} (q4)

				(q3) edge  [anchor=south] node {$a$} (q5_1)
				(q4) edge  [anchor=south] node {$b$} (q5_1)

				(q5_1) edge  [] node {$a,b$} (q6_1)
				(q5_2) edge  [] node {$a,b$} (q6_2)

				(q1) edge [loop left] node {$c$} (q1)
				(q1) edge [bend right] node[left=4pt] {$a$} (q5_2)
			;
		\end{tikzpicture}
	\end{minipage}
	\vskip 1em
	\begin{tikzpicture}[->,>=stealth',shorten >=1pt,auto,node distance=2cm,thick,scale=0.9,main node/.style={circle,fill=blue!20,draw,font=\sffamily}]
		\node[initial,initial text=,state] (q0) {$q_I$};
		\node[state] (p0) at ($(q0) + (-2.2,-1.5) $) 	{$q_0$};
		\node[state] (p1) [below left of=p0] 	{$q_1$};
		\node[state] (p2) [above left of=p1] 	{$q_2$};

		\node[state] (q2) at ($(q0) + (2.2,-1.5) $)	{};

		\node[state] (q3) [below left of=q2]	{};
		\node[state,accepting] (q4) [below right of=q2]	{};

		\node[state,accepting] (q5_1) [below right of=q3]	{};
		\node[state,accepting] (q5_2) at ($(q5_1)+(-3,0)$)	{};
		\node[state,accepting] (q5_3) at ($(q5_2)+(-3,0)$)	{};
		\node[state,accepting] (q5_4) at ($(q5_3)+(-3,0)$)		{};

		\node[accepting,state] (q6_1) [below of=q5_1]	{};
		\node[accepting,state] (q6_2) [below of=q5_2]	{};
		\node[accepting,state] (q6_3) [below of=q5_3]	{};
		\node[accepting,state] (q6_4) [below of=q5_4]	{};

		\path[every node/.style={font=\sffamily\small}]
			(q0) edge  [anchor=south] node {$a$} (p0)
			(q0) edge  [anchor=south] node {$c$} (q2)

			(q2) edge  [anchor=south] node {$b$} (q3)
			(q2) edge  [anchor=south] node {$c$} (q4)

			(q3) edge  [anchor=south] node {$a$} (q5_1)
			(q4) edge  [anchor=south] node {$b$} (q5_1)

			(q5_1) edge  [anchor=east] node {$a$} (q6_1)
			(q5_1) edge  [below] node {$b$} (q6_2)
			(q5_2) edge  [anchor=east] node {$a$} (q6_2)
			(q5_2) edge  [above] node {$b$} (q6_1)
			(q5_3) edge  [anchor=west] node {$a,b$} (q6_3)
			(q5_4) edge  [anchor=west] node {$a,b$} (q6_4)

			(p0) edge [bend left] node {$c$} (p1)
			(p1) edge [bend left] node {$c$} (p2)
			(p2) edge [bend left] node {$c$} (p0)

			(p0) edge [] node {$a$} (q5_2)
			(p1) edge [] node {$a$} (q5_3)
			(p2) edge [] node {$a$} (q5_4)
		;
	\end{tikzpicture}
	\vskip 1em
	\caption{The minimum \dfa\ accepting the language $L=cc+(ac^*a+c(cb+ba))(\varepsilon+a+b)$, an equivalent minimal \revdfa, and an equivalent reduced \revdfa}
	\label{fig:min_dfa2}
\end{figure}

In Theorem~\ref{th:irrev} we presented a sufficient condition for the existence of infinitely many reduced \revdfas\ accepting a given language. 
The condition requires the existence (in the irreversible part of the minimum \dfa) of a state~$s$ entered by a transition~$a$, on a path starting in a loop, and by another symbol~$b$, different from~$a$, such that the symbol~$b$ has a ``double'' transition entering in~$s$ in the following sense: either~$M$ contains two transitions on~$b$ entering in~$s$, or the only transition on~$b$ entering in~$s$ is from another state~$r$ in the irreversible part, thus implying that in each equivalent \revdfa\ this transition should be duplicated (with the two states~$r$ and~$s$).

These situations are presented in Figure~\ref{fig:min_dfa1} and Figure~\ref{fig:min_dfa2}. 
In particular, in the minimum \dfa\ in Figure~\ref{fig:min_dfa1} there are two transitions on the letter~$b$ entering in the state~$s$, while in the minimum \dfa\ in Figure~\ref{fig:min_dfa2}, the state $r$ with $\delta(r,b)=s$ is in the irreversible part.
Figure~\ref{fig:min_dfa3} instead shows a minimum \dfa\ which does not satisfy such sufficient condition: the state~$s$ is the unique in the irreversible part and~$\#\delta^R(s,b)=1$.
If we try to expand the loop, we can reduce the obtained automata to the minimal one.
Hence, for this example, it is not possible to build infinitely many reduced \revdfas.

We point out that the condition in Theorem~\ref{th:irrev} is not necessary. 
In fact it  is possible to build infinitely many reduced \revdfas\ accepting a same language even in other situations~\cite{SI16}.

All three examples show a minimum \dfa\ containing a loop in the reversible part.
Using Theorem~\ref{th:irrev} we now prove that when the loop is in the irreversible part, it is always possible to construct an infinite family of reduced \revdfas\ (for an example see Figure~\ref{fig:red}).

\begin{corollary}
	\label{cor:irrev-inf}
	Let~$M=(Q,\Sigma,\delta,q_I,F)$ be the minimum \dfa\ accepting a reversible language~$L$.
	If~$M$ contains a loop in the irreversible part, then there exist infinitely many nonisomorphic reduced \revdfas\ accepting~$L$.
\end{corollary}
\begin{proof}
Let us consider a nontrivial \scc\ $C$ in the irreversible part of $M$. As observed in Section~\ref{sec:reversible}, $C$ should be different from the \scc\ $C_{q_I}$ containing the initial state. Then there exist a string $y\in \Sigma^*$, a symbol $b\in\Sigma$, a state $r$ not in $C$ and a state $s$ in C such that $\delta(q_I,y)= r$ and $\delta(r,b)=s$, i.e., the \scc\ $C$ is reached from the initial state after reading the whole string $yb$. Since $C$ is in the irreversible part, we can find $y,b,r,s$ in such a way that either $c(r)>1$ or $\#\delta^R(s,b)>1$.

Furthermore, since $C$ is nontrivial there should exist a string $x\in\Sigma^+$ such that $\delta(s,x)=s$. Let $a$ be the last symbol of $x$. Then $a\neq b$, otherwise $M$ should contain the forbidden pattern. This allows to conclude that $M$ satisfied the condition of Theorem~\ref{th:irrev}. Hence, there are infinitely many reduced automata equivalent to $M$.

\end{proof}

\begin{figure}%
	\centering
	\begin{minipage}{.5\textwidth}
		\centering
		\begin{tikzpicture}[->,>=stealth',shorten >=1pt,auto,node distance=2cm,
	thick,main node/.style={circle,fill=blue!20,draw,font=\sffamily}]

			\node[initial,initial text=,state] (q0) {$q_I$};
			\node[state] (q1) [below left of=q0]	{};
			\node[state] (q2) [below right of=q0]	{};

			\node[state] (q3) [below left of=q2]	{};
			\node[state,accepting] (q4) [below right of=q2]	{};

			\node[state,accepting] (q5) [below left of=q4]	{$s$};

			\path[every node/.style={font=\sffamily\small}]
				(q0) edge  [] node[anchor=south] {$a$} (q1)
				(q0) edge  [] node[anchor=south] {$c$} (q2)

				(q2) edge  [] node[anchor=south] {$b$} (q3)
				(q2) edge  [] node[anchor=south] {$c$} (q4)

				(q3) edge  [] node[anchor=south] {$a$} (q5)
				(q4) edge  [] node[anchor=south] {$b$} (q5)

				(q1) edge [loop left] node {$c$} (q1)
				(q1) edge [bend right] node[left] {$a$} (q5);
		\end{tikzpicture}
	\end{minipage}%
	\begin{minipage}{.5\textwidth}
		\centering
		\begin{tikzpicture}[->,>=stealth',shorten >=1pt,auto,node distance=2cm,
	thick,main node/.style={circle,fill=blue!20,draw,font=\sffamily}]

			\node[initial,initial text=,state] (q0) {$q_I$};
			\node[state] (q1) [below left of=q0]	{};
			\node[state] (q2) [below right of=q0]	{};

			\node[state] (q3) [below left of=q2]	{};
			\node[state,accepting] (q4) [below right of=q2]	{};

			\node[state,accepting] (q5) [below left of=q4]	{};
			\node[state,accepting] (q5bis) at (q1|-q5)	{};

			\path[every node/.style={font=\sffamily\small}]
				(q0) edge  [] node[anchor=south] {$a$} (q1)
				(q0) edge  [] node[anchor=south] {$c$} (q2)

				(q2) edge  [] node[anchor=south] {$b$} (q3)
				(q2) edge  [] node[anchor=south] {$c$} (q4)

				(q3) edge  [] node[anchor=south] {$a$} (q5)
				(q4) edge  [] node[anchor=south] {$b$} (q5)

				(q1) edge [loop left] node {$c$} (q1)
				(q1) edge [left] node[] {$a$} (q5bis);
		\end{tikzpicture}
	\end{minipage}
	\vskip1em
\caption{The minimum \dfa\ accepting the language $L=ac^*a+ c(c+cb+ba)$, with an equivalent minimal \revdfa}
\label{fig:min_dfa3}
\end{figure}
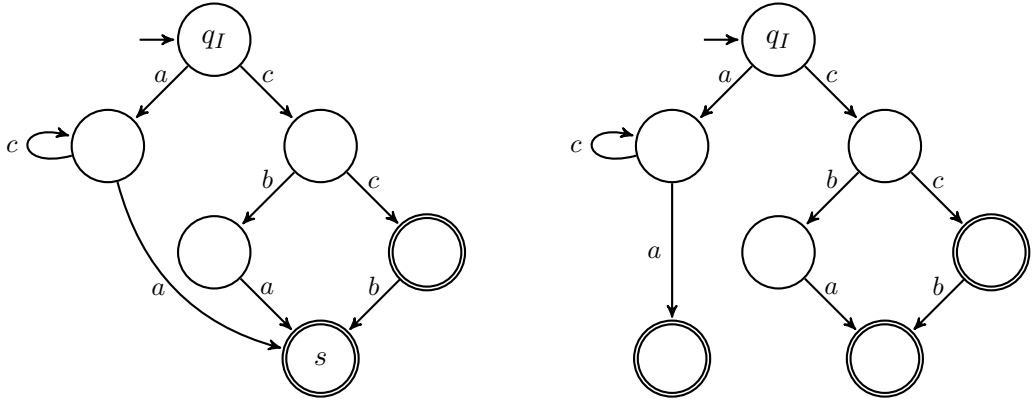

\section{Unique Minimal Reversible Automata}
\label{sec:unique}
In this section we resume the investigation about the existence of a unique minimal \revdfa\ (cfr. Section~\ref{sec:minimal}), by proving that when the family of \revdfas\ accepting a given language~$L$ contains a unique minimal \revdfa~$A$ then each equivalent \revdfa\ can be reduced to it. As a consequence, $A$ is the \emph{minimum} \revdfa\ accepting~$L$ and no reduced \revdfas\ other than~$A$ can exist.

We remind the reader that a reversible language has a unique minimal reversible automaton if and only if all the transitions entering in a state in the irreversible part of the minimum \dfa\ are on the same symbol (Theorem~\ref{th:minimal}).

\begin{lemma}
	\label{lemma:copies}
	Let~$M=(Q,\Sigma,\delta,q_I,F)$ be the minimum \dfa\ accepting a reversible language~$L$ such that, for each state~$q\in Q$ with $c(q)>1$, all incoming transitions to~$q$ are on the same symbol, that is~$\#\{a \in \Sigma \mid \delta^R(q,a)\neq\emptyset\}=1$.
	Let~$\mathcal{W}_q$ be the set of pairs~$(r,x)$ where~$r$ is a reversible state and~$x\in\Sigma^*$ such that, starting from~$r$ and reading~$x$, the state~$q$ is reached, passing only through states in the irreversible part, i.e.,
	$$\mathcal{W}_q=\{(r,x) \mid \delta(r,x)=q,\, c(r)=1,\, \forall x=uv,\,u\in\Sigma^+,v\in\Sigma^* \implies c(\delta(r,u))>1\}\text.$$
	Then the minimum number of states equivalent to~$q$ in the minimal \revdfa\ accepting~$L$ is equal to the cardinality of~$\mathcal{W}_q$, i.e., $c(q)=\#\mathcal{W}_q\text.$ 
\end{lemma}
\begin{proof}
	Let~$A=(Q_A,\Sigma,\delta_A,q_{I_A},F_A)$ be a minimal \revdfa~equivalent to~$M$ with the morphism~$\varphi:A\to M$.
	We consider all the paths in~$A$ from states~$p$, where~$\varphi(p)=r$ is reversible, reading a string~$x$ such that~$(r,x)\in\mathcal{W}_q$ and reaching some state in~$\varphi^{-1}(q)$.
	If two such paths share a same state~$t$, other than the initial state~$p$ of the path, then at least two transitions enter in~$t$.
	Since~$A$ is reversible, such transitions should be on different letters.
	However, this is not possible because of the hypothesis on the incoming transitions on a same symbol to the states in the irreversible part of~$M$.
	Furthermore, we notice that each state in~$\varphi^{-1}(q)$ should be reached by some of those paths.
	Hence we conclude that~$c(q)=\#\mathcal{W}_q$.
\end{proof}

An example related to Lemma~\ref{lemma:copies} is shown in Figure~\ref{fig:copies}: considering the minimum \dfa~$M$ (on the top left), we can observe that~$\mathcal{W}_q=\{(r_1,ba),(r_2,ba)\}$. So, the sufficient and necessary number of copies of~$q$ is~$c(q)=\#\mathcal{W}_q=2$, as shown in the equivalent minimal \revdfa\ (on the top right).

\begin{figure}[htp]
	\begin{minipage}{.5\textwidth}
		\centering
		\begin{tikzpicture}[->,>=stealth',shorten >=1pt,auto,node distance=2cm,thick,main node/.style={circle,fill=blue!20,draw,font=\sffamily}]
			\node[initial,initial text=,state] (q0) {$q_I$};
			\node[state] (q1) [right of=q0]	{};
			\node[state] (q2) [below of=q0]	{$r_1$};
			\node[state] (q3) [below of=q2]	{$r_2$};
			\node[state] (q4) [right of=q3]	{};
			\node[state,accepting] (q5) [right of=q4]	{$q$};
			\path[every node/.style={font=\sffamily\small}]
				(q0) edge  [] node[] {$b$} (q1)
				(q0) edge  [left] node[] {$a$} (q2)
				(q1) edge  [] node[] {$b$} (q2)
				(q2) edge  [left] node[] {$a$} (q3)
				(q2) edge  [] node[] {$b$} (q4)
				(q3) edge  [] node[] {$b$} (q4)

				(q4) edge  [] node[] {$a$} (q5)
			;
		\end{tikzpicture}
	\end{minipage}%
	\begin{minipage}{.5\textwidth}
		\centering
		\begin{tikzpicture}[->,>=stealth',shorten >=1pt,auto,node distance=2cm,thick,main node/.style={circle,fill=blue!20,draw,font=\sffamily}]
			\node[initial,initial text=,state] (q0) {$q_I$};
			\node[state] (q1) [right of=q0]	{};
			\node[state] (q2) [below of=q0]	{$p_1$};
			\node[state] (q3) [below of=q2]	{$p_2$};

			\node[state] (q4) [right of=q2]	{};
			\node[state,accepting] (q5) [right of=q4]	{$q_1$};
			\node[state] (q4bis) [right of=q3]	{};
			\node[state,accepting] (q5bis) [right of=q4bis]	{$q_2$};

			\path[every node/.style={font=\sffamily\small}]
				(q0) edge  [] node [] {$b$} (q1)
				(q0) edge  [left] node [] {$a$} (q2)
				(q1) edge  [] node [] {$b$} (q2)
				(q2) edge  [left] node [] {$a$} (q3)

				(q2) edge  [] node [] {$b$} (q4)
				(q3) edge  [] node [] {$b$} (q4bis)

				(q4) edge  [] node [] {$a$} (q5)
				(q4bis) edge  [] node [] {$a$} (q5bis)
			;
		\end{tikzpicture}
	\end{minipage}
	\vskip 1em
	\centering
	\begin{tikzpicture}[->,>=stealth',shorten >=1pt,auto,node distance=2cm,thick,main node/.style={circle,fill=blue!20,draw,font=\sffamily}]
		\node[initial,initial text=,state] (q0) {$q_I$};
		\node[state] (q1) [right of=q0]	{};
		\node[state] (q2) [below of=q0]	{$p_1'$};
		\node[state] (q2bis) [right of=q2]	{$p_1''$};
		\node[state] (q3) [below of=q2]	{$p_2'$};
		\node[state] (q3bis) [right of=q3]	{$p_2''$};

		\node[state] (q4) [right of=q2bis]	{};
		\node[state,accepting] (q5) [right of=q4]	{$q_1''$};
		\node[state] (q4bis) [right of=q3bis]	{};
		\node[state,accepting] (q5bis) [right of=q4bis]	{$q_2''$};
		\node[state] (q4ter) [below of=q4bis]	{};
		\node[state,accepting] (q5ter) [right of=q4ter]	{$q_2'$};
		\node[state] (q4qtr) [below of=q4ter]	{};
		\node[state,accepting] (q5qtr) [right of=q4qtr]	{$q_1'$};

		\path[every node/.style={font=\sffamily\small}]
			(q0) edge  [] node [] {$b$} (q1)
			(q0) edge  [anchor=south east] node [] {$a$} (q2)
			(q1) edge  [anchor=south west] node [] {$b$} (q2bis)
			(q2) edge  [left] node [] {$a$} (q3)

			(q2) edge  [bend right,in=-130,out=-70, left] node [] {$b$} (q4qtr)
			(q2bis) edge  [] node [] {$b$} (q4)
			(q2bis) edge  [] node [] {$a$} (q3bis)
			(q3) edge  [bend right, left] node [] {$b\;$} (q4ter)
			(q3bis) edge  [] node [] {$b$} (q4bis)

			(q4) edge  [] node [] {$a$} (q5)
			(q4bis) edge  [] node [] {$a$} (q5bis)
			(q4ter) edge  [] node [] {$a$} (q5ter)
			(q4qtr) edge  [] node [] {$a$} (q5qtr)
		;
	\end{tikzpicture}
	\vskip1em
	\caption{The minimum \dfa\ accepting the language $L=(bb+a)(ab+b)a$ (top left), with an equivalent minimal \revdfa\ (top right) and an equivalent non minimal \revdfa\ (bottom). Notice that the minimum \dfa\ (on the top left) accepting the language is different from the minimum \revdfa}
	\label{fig:copies}
\end{figure}
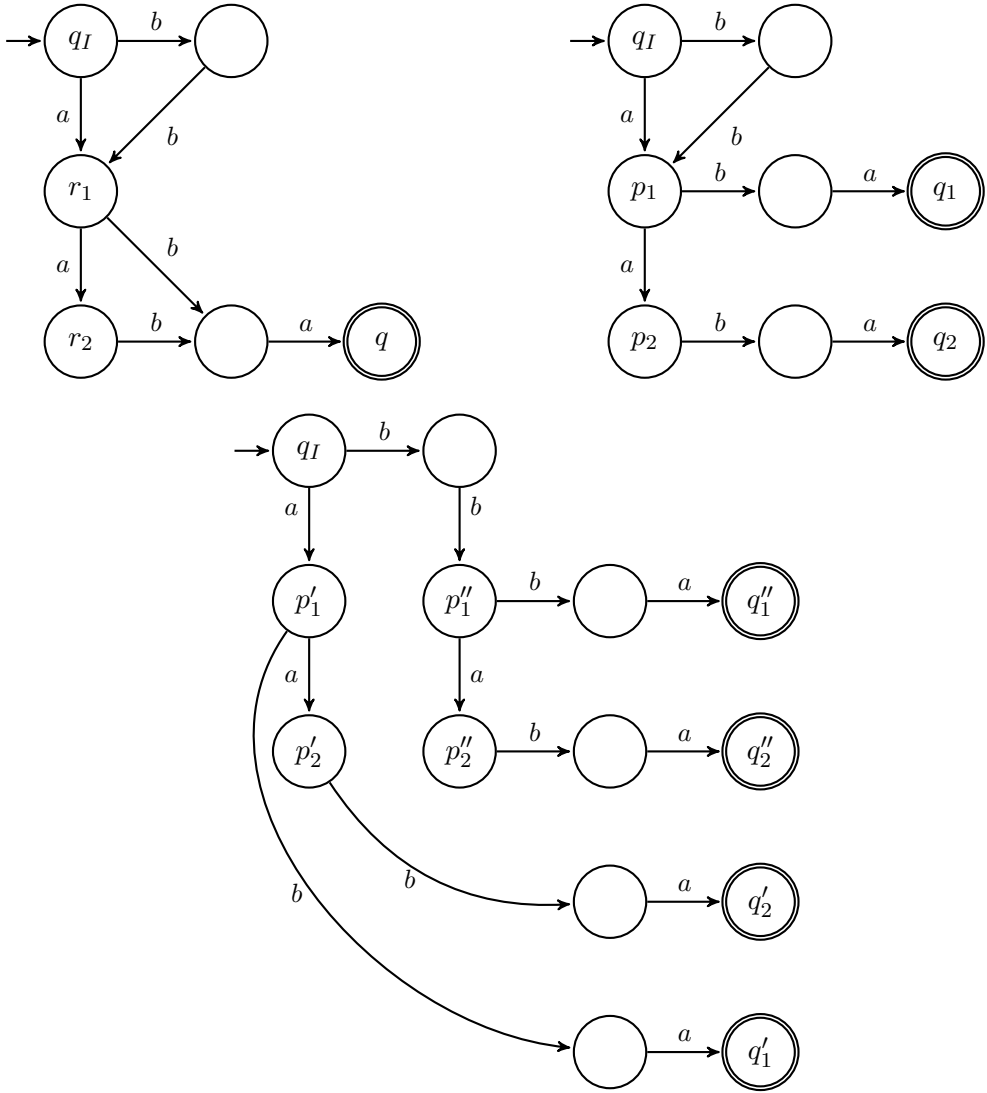

We are now ready to prove the following:

\begin{theorem}
\label{th:unique-minimal}
Let $L$ be a reversible language and $M=(Q,\Sigma,\delta,q_I,F)$ be the minimum \dfa\ accepting it. If there exists a unique minimal \revdfa\  $A_m$ accepting $L$, then each \revdfa\ $A$ other than $A_m$ accepting $L$ is not reduced.
\end{theorem}
\begin{proof}
From Corollary~\ref{cor:unique-minimal} it follows that there are no loops in the irreversible part of $M$ and all incoming transitions to each state $q\in Q$ with $c(q)>1$ are on the same symbol. 
As a consequence, also in any \dfa\ equivalent to~$M$ all incoming transitions to  any state which is equivalent to $q$ are on the same symbol. 

Let $A= (Q_A,\Sigma,\delta_A,q_{I_A},F_A)$ be a \revdfa\ accepting $L$ with the morphism $\varphi:A\to M$. Suppose $A$ is different from  $A_m$. Then $A$ is not minimal and, as a consequence of Lemma~\ref{lemma:morph}, it should exist $q\in Q$ such that $\#\varphi{^{-1}}(q) > c(q)$.  We are going to prove that from $A$ it is possible to obtain a smaller \revdfa\ accepting~$L$.
\begin{enumerate}

\item\label{B} If $q$ is in the reversible part, i.e., $c(q) = 1$, 
from~$A$ we obtain a smaller \dfa~$B=(Q_B,\Sigma,\delta_B, q_{I_B}, F_B)$ by merging together all states which are equivalent to a same state belonging to the reversible part of~$M$ (which would imply to merge also all the states that are reachable reading a same string from states that have been merged). Hence, the \dfa~$B$ so obtained consists of a copy of the reversible part of~$M$ together with some other states and transitions which are derived from~$A$ and which correspond to the irreversible part of~$M$.

Suppose that~$B$ contains an irreversible state $p_B$, namely there exist two entering transitions in $p_B$ with a same symbol $a$ from two different states, say $r'_B, r''_B\in Q_B$.
This should imply that $p_B$ is equivalent to a state $p\in Q$ in the irreversible part of $M$. 
So,~$p_B$ has been obtained merging some equivalent states $p_1,p_2,\ldots,p_k\in Q_A$, where $\varphi(p_i)=p$, $i=1,\ldots, k$.
Since the \emph{only} way to enter in~$p_1,p_2,\ldots,p_k$ is with the symbol~$a$, this happens if and only if there exist states $r_1,r_2,\ldots,r_k \in Q_A$, with $\delta(r_i,a)=p_i$, $i=1,\ldots,k$, which have been merged in a state 
$r_B\in Q_B$, obtaining~$\delta_B(r_B,a)=p_B$. Since no other transitions can enter in~$p_1,p_2,\ldots,p_k$, it follows that $r'_B=r''_B=r_B$, which is a contradiction to the hypothesis that~$p_B$ is irreversible.
Thus, we conclude that~$B$ is reversible and it is smaller than $A$.

\item\label{A} 
If $q$ is in the irreversible part, i.e., $c(q)>1$, let~$\varphi^{-1}(q)= \{q_1,\ldots, q_k\}$ with~$k=\#\varphi^{-1}(q)>c(q)$. 
Let $\mathcal{W}_q$ be the set defined in Lemma~\ref{lemma:copies}. 
There exist couples~$(r_1,x_1), \ldots, (r_k, x_k)\in \mathcal{W}_q$ and $p_1,\ldots, p_k \in Q_A$ such that $r_i= \varphi(p_i)$, and $\delta_A(p_i, x_i)=q_i$, $i=1,\ldots,k$.
By Lemma~\ref{lemma:copies}, $\#\mathcal{W}_q=c(q)$ so, since $k>c(q)$, two of such couples coincide, i.e., there are two indices $i,j$, $i\neq j$, such that $(r_i,x_i)=(r_j,x_j)$. If $p_i = p_j$ then it should be $\delta_A(p_i,x_i)= \delta_A(p_j,x_j)$, that is $q_i=q_j$, which is a contradiction to the hypothesis that $\#\varphi^{-1}(q)=k$. Thus, $p_i \neq p_j$, but $\varphi(p_i)=r_i=r_j=\varphi(p_j)$.  
Hence, $\#\varphi^{-1}(r_i)>1$, where~$r_i$ is in the reversible part of $M$. This brings us to case~(\ref{B}).

\end{enumerate}
As a consequence, the automaton~$A$ is not reduced.
\end{proof}

By Theorem~\ref{th:unique-minimal}, when a language~$L$ is accepted by a unique minimal \revdfa~$A_m$, then any other \revdfa~$A$ accepting~$L$ can be reduced, by merging states, to~$A_m$, that is there exists a morphism from~$A$ to~$A_m$. Hence, $A_m$ is the \emph{minimum} \revdfa\ accepting~$L$. Notice that~$A_m$ could differ from the minimum \dfa\ accepting~$L$.
This is illustrated (together some of the ideas in the proofs of Lemma~\ref{lemma:copies} and Theorem~\ref{th:unique-minimal}) in Figure~\ref{fig:copies}.

\section{Conclusion}
\label{sec:conclu}
We studied the existence of many minimal and reduced \revdfas. 
In some cases the minimum \dfa\ accepting a language is already reversible, so assuring that the language is reversible. However, in general a minimum \dfa\ does not need to be reversible, although the accepted language could be reversible. Using Theorem~\ref{th:condition} and the construction from~\cite{HJK15} outlined in Section~\ref{sec:prel}, in the case the language is reversible, from a given minimum \dfa\ we can obtain an equivalent minimal \revdfa.

Minimal \revdfas\  are not necessarily unique (see Figure~\ref{fig:min} for an example, while Figure~\ref{fig:min_isomorphic} shows a case with a unique minimal \revdfa).
In Section~\ref{sec:minimal} we gave a characterization of the languages having a unique minimal \revdfa, in terms of the structure of the minimum \dfa.
As shown in Section~\ref{sec:unique}, if a language has a unique minimal \revdfa, then each equivalent \revdfa\ can be reduced to it, namely the minimal \revdfa\ is the minimum \revdfa\ for the language under consideration.

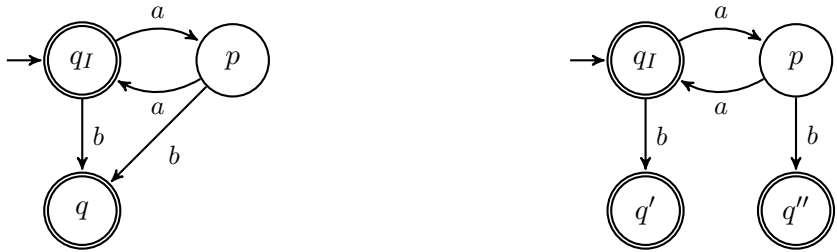
\begin{figure}[ht]
	\begin{minipage}{.5\textwidth}
		\centering
		\begin{tikzpicture}[->,>=stealth',shorten >=1pt,auto,node distance=2cm,
thick,main node/.style={circle,fill=blue!20,draw,font=\sffamily},main node2/.style={circle,fill=green!20,draw,font=\sffamily},main node3/.style={circle,fill=yellow!20,draw,font=\sffamily}]

			\node[accepting,initial,initial text=,state] (q0) {$q_I$};
			\node[state] (q1) [right  of=q0]	{$p$};
			\node[accepting,state] (q2) [below of=q0]	{$q$};

			\path[every node/.style={font=\sffamily\small}]
				(q0) edge  [bend left] node {$a$} (q1)
				(q1) edge  [bend left] node {$a$} (q0)

				(q0) edge [] node {$b$} (q2)
				(q1) edge [] node {$b$} (q2);
		\end{tikzpicture}
	\end{minipage}%
	\begin{minipage}{.5\textwidth}
		\centering
		\begin{tikzpicture}[->,>=stealth',shorten >=1pt,auto,node distance=2cm,
thick,main node/.style={circle,fill=blue!20,draw,font=\sffamily},main node2/.style={circle,fill=green!20,draw,font=\sffamily},main node3/.style={circle,fill=yellow!20,draw,font=\sffamily}]

			\node[accepting,initial,initial text=,state] (q0) {$q_I$};
			\node[state] (q1) [right  of=q0]	{$p$};
			\node[accepting,state] (q21) [below of=q0]	{$q'$};
			\node[accepting,state] (q22) [below of=q1]	{$q''$};

			\path[every node/.style={font=\sffamily\small}]
				(q0) edge  [bend left] node {$a$} (q1)
				(q1) edge  [bend left] node {$a$} (q0)

				(q0) edge [] node {$b$} (q21)
				(q1) edge [] node {$b$} (q22);
		 \end{tikzpicture}
	\end{minipage}
	\caption{The minimum \dfa\ and the minimum \revdfa\ accepting the language $L=(aa)^*+a^*b$}
	\label{fig:min_isomorphic}
\end{figure}

When the minimal \revdfa\ is not unique, reduced not minimal \revdfas\ could exist. In Theorem~\ref{th:irrev} we gave a sufficient condition for the existence of arbitrarily large and, hence, infinitely many reduced \revdfas\ accepting a same reversible language. However, there are also languages having only finitely many reduced non minimal \revdfas. An example is given in Figure~\ref{fig:nonmin_reduced}, where  the accepted language is finite, thus implying that only finitely many automata with useful states and, hence, finitely many reduced \revdfas, can exist.

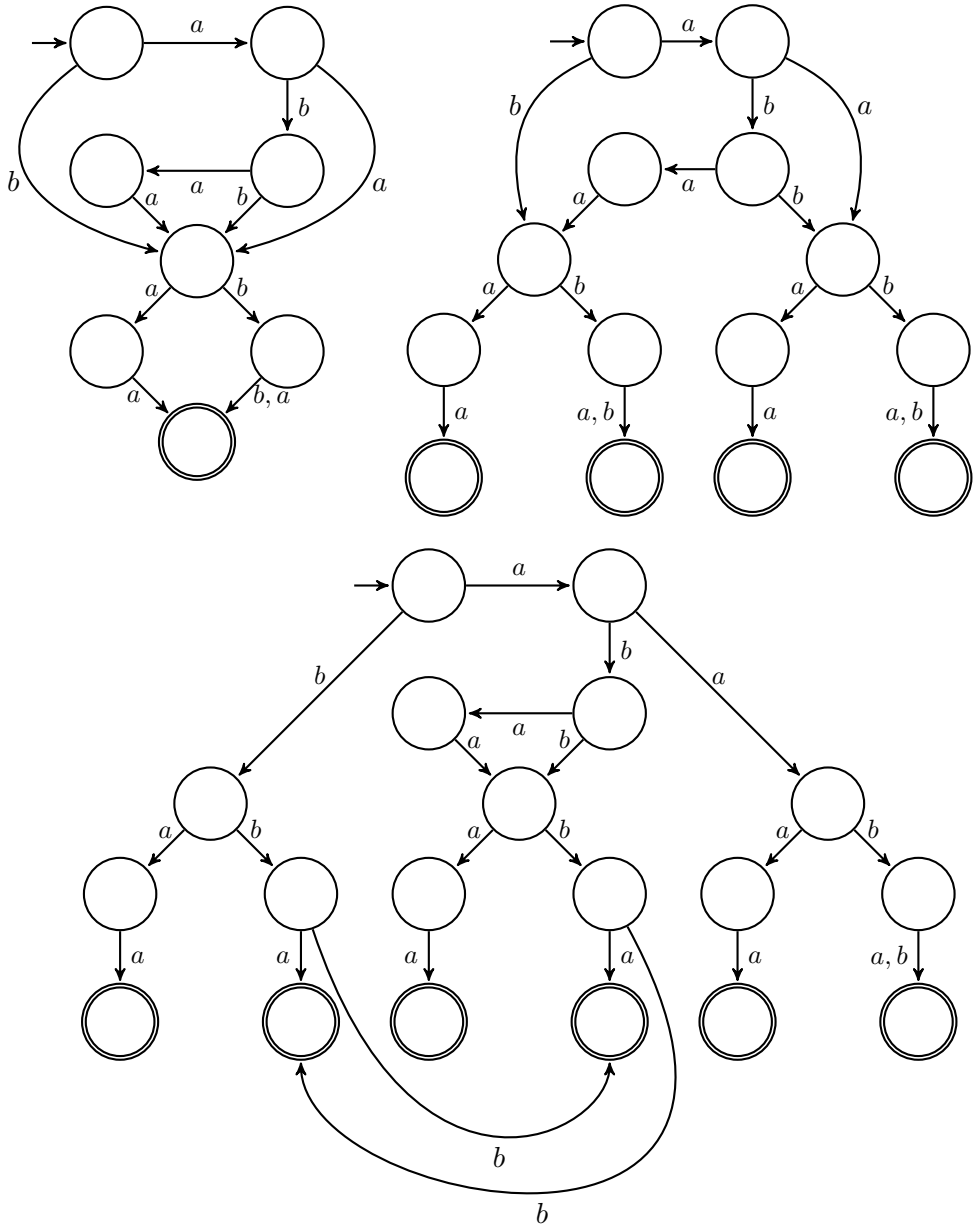
\begin{figure}[H]
	\begin{minipage}{.3\textwidth}
		\centering
		\begin{tikzpicture}[->,>=stealth',shorten >=1pt,auto,node distance=1.7cm,thick]
		\node[state] (q4) {};
		\node[state] (q2) [above left of=q4]	{};
		\node[state] (q3) [above right of=q4]	{};
		\node[initial,initial text=,state] (q0) [above of=q2] {};
		\node[state] (q1) [above of=q3]	{};
		\node[state] (q5) [below left of=q4]	{};
		\node[state] (q6) [below right of=q4]	{};
		\node[state, accepting] (q7) [below right of=q5]	{};
		\node[state,draw=none] () [below of=q5] {};

		\path[every node/.style={font=\sffamily\small}]
			(q0) edge  [] node {$a$} (q1)
			(q1) edge  [] node {$b$} (q3)
			(q3) edge  [] node {$a$} (q2)

			(q2) edge [above] node {$a$} (q4)
			(q3) edge [above] node {$b$} (q4)

			(q4) edge [above] node {$a$} (q5)
			(q4) edge [above] node {$b$} (q6)
			(q5) edge [left] node {$a$} (q7)
			(q6) edge [right] node {$b,a$} (q7);

		\draw[->] (q0) .. controls ($(q0)+(-2cm,-1.5cm)$) and ($(q4)+(-2cm,0.5cm)$) .. node [left,pos=0.497] {$b$} (q4);
		\draw[->] (q1) .. controls ($(q1)+(2cm,-1.5cm)$) and ($(q4)+(2cm,0.5cm)$) .. node [right,pos=0.5143] {$a$} (q4);

		 \end{tikzpicture}
	\end{minipage}%
	\begin{minipage}{.7\textwidth}
		\centering
		\begin{tikzpicture}[->,>=stealth',shorten >=1pt,auto,node distance=1.7cm,thick]
			\node[initial,initial text=,state] (q0) [above of=q2] {};
			\node[state] (q1) [right of=q0]	{};
			\node[state] (q2) [below of=q0]	{};
			\node[state] (q3) [below of=q1]	{};

			\node[state] (q4) [below left of=q2] {};
			\node[state] (q5) [below left of=q4]	{};
			\node[state] (q6) [below right of=q4]	{};
			\node[state, accepting] (q7) [below of=q5]	{};
			\node[state, accepting] (q7bis) [below of=q6]	{};

			\node[state] (q4bis) [below right of=q3] {};
			\node[state] (q5bis) [below left of=q4bis]	{};
			\node[state] (q6bis) [below right of=q4bis]	{};
			\node[state, accepting] (q7ter) [below of=q5bis]	{};
			\node[state, accepting] (q7quater) [below of=q6bis]	{};

			\path[every node/.style={font=\sffamily\small}]
				(q0) edge  [] node {$a$} (q1)
				(q1) edge  [] node {$b$} (q3)
				(q3) edge  [] node {$a$} (q2)

				(q2) edge [anchor=south] node {$a$} (q4)
				(q3) edge [anchor=south] node {$b$} (q4bis)

				(q4) edge [anchor=south] node {$a$} (q5)
				(q4) edge [anchor=south] node {$b$} (q6)
				(q5) edge [anchor=west] node {$a$} (q7)
				(q6) edge [anchor=east] node {$a,b$} (q7bis)

				(q4bis) edge [anchor=south] node {$a$} (q5bis)
				(q4bis) edge [anchor=south] node {$b$} (q6bis)
				(q5bis) edge [anchor=west] node {$a$} (q7ter)
				(q6bis) edge [anchor=east] node {$a,b$} (q7quater);

			\draw[->] (q0) .. controls ($(q0)+(-1cm,-0.5cm)$) and ($(q4)+(-0.5cm,2cm)$) .. node [left,pos=0.497] {$b$} (q4);
			\draw[->] (q1) .. controls ($(q1)+(1cm,-0.5cm)$) and ($(q4bis)+(0.5cm,2cm)$) .. node [right,pos=0.5143] {$a$} (q4bis);
		\end{tikzpicture}
	\end{minipage}
	\vskip1em
	\centering
	\begin{tikzpicture}[->,>=stealth',shorten >=1pt,auto,node distance=1.7cm,thick]
		\node[state] (q4) {};
		\node[state] (q2) [above left of=q4]	{};
		\node[state] (q3) [above right of=q4]	{};
		\node[initial,initial text=,state] (q0) [above of=q2] {};
		\node[state] (q1) [above of=q3]	{};
		\node[state] (q5) [below left of=q4]	{};
		\node[state] (q6) [below right of=q4]	{};
		\node[state, accepting] (q7) [below of=q5]	{};
		\node[state, accepting] (q8) [below of=q6]	{};

		\node[state] (q6bis) [left of=q5]	{};
		\node[state, accepting] (q8bis) [left of=q7]	{};
		\node[state] (q4bis) [above left of=q6bis] {};
		\node[state] (q5bis) [below left of=q4bis]	{};
		\node[state, accepting] (q7bis) [below of=q5bis]	{};

		\node[state] (q5ter) [right of=q6]	{};
		\node[state, accepting] (q7ter) [right of=q8]	{};
		\node[state] (q4ter) [above right of=q5ter] {};
		\node[state] (q6ter) [below right of=q4ter]	{};
		\node[state, accepting] (q8ter) [below of=q6ter]	{};

		\path[every node/.style={font=\sffamily\small}]
			(q0) edge  [] node {$a$} (q1)
			(q1) edge  [] node {$b$} (q3)
			(q3) edge  [] node {$a$} (q2)

			(q2) edge [above] node {$a$} (q4)
			(q3) edge [above] node {$b$} (q4)

			(q0) edge  [anchor=south] node {$b$} (q4bis)
			(q1) edge  [anchor=south] node {$a$} (q4ter)

			(q4) edge [anchor=south] node {$a$} (q5)
			(q4) edge [anchor=south] node {$b$} (q6)
			(q5) edge [anchor=east] node {$a$} (q7)
			(q6) edge [anchor=west] node {$a$} (q8)

			(q4bis) edge [anchor=south] node {$a$} (q5bis)
			(q4bis) edge [anchor=south] node {$b$} (q6bis)
			(q5bis) edge [anchor=west] node {$a$} (q7bis)
			(q6bis) edge [anchor=east] node {$a$} (q8bis)

			(q4ter) edge [anchor=south] node {$a$} (q5ter)
			(q4ter) edge [anchor=south] node {$b$} (q6ter)
			(q5ter) edge [anchor=west] node {$a$} (q7ter)
			(q6ter) edge [anchor=east] node {$a,b$} (q8ter);

	\draw[->]
					(q6) ..  controls  ($(q6)+(3cm,-5.5cm)$) and ($(q8bis)+(0,-2.5cm)$) .. node[align=center,below] {$b$} (q8bis);
	\draw[->]
					(q6bis) ..  controls  ($(q6bis)+(1.5cm,-4.5cm)$) and ($(q8)+(0,-1.5cm)$) .. node[align=center,below] {$b$} (q8);

	\end{tikzpicture}
	\vskip1em
\caption{The minimum \dfa\ accepting the language $L=(a(bb+a+baa) + b)(aa+ba+bb)$, an equivalent minimal \revdfa, and an equivalent reduced non minimal \revdfa
}
\label{fig:nonmin_reduced}
\end{figure}

\bibliographystyle{alpha}%
	\bibliography{biblio}

\end{document}